\documentclass[a4paper,11pt]{article}

\usepackage{amsmath,amsthm,amsfonts,amssymb,bbm,xcolor,mathtools}
\usepackage[margin=2cm]{geometry}
\usepackage{etoolbox}

\usepackage{svg}
\usepackage{tikz}
\usepackage[only,bigsqcap]{stmaryrd}
\usepackage{subfig}

\usepackage{hyperref}
\usepackage{cleveref}

\usepackage[shortlabels]{enumitem}

\theoremstyle{plain}
\newtheorem{theorem}{Theorem}[section]

\newtheorem{corollary}  [theorem]{Corollary}

\newtheorem{example}    [theorem]{Example}
\newtheorem{lemma}      [theorem]{Lemma}

\newtheorem{proposition}[theorem]{Proposition}

\theoremstyle{definition}

\newtheorem*{claim*}{Claim}

\newcommand{\N} {\mathbb{N}}

\newcommand{\B} {\mathbb{B}}

\newcommand{\Define}    [1] {\textbf{#1}}

\newcommand{\Functions}             { \mathrm{F} }

\newcommand{\Updates}               [1] { \mathtt{#1} }

\newcommand{\Asynchronous}              { \Updates{A} }

\newcommand{\TrappingGraph}             { \Updates{T} }

\newcommand{\History}                   { \Updates{H} }

\newcommand{\MostPermissive}            { \Updates{M} }

\newcommand{\Subcube}         { \Updates{S} }

\newcommand{\Cuttable}                  { \Updates{C} }
\newcommand{\Interval}                  { \Updates{I} }

\newcommand{\Trapping}              [1] { {#1}^\mathrm{T} }

\newcommand{\HammingDistance}{d_\mathrm{H}}

\newcommand{\Network}[1]{ \mathsf{#1} }

\newcommand{\AsynchronousNetwork}{ \Network{A} }
\newcommand{\HistoryNetwork}{ \Network{H} }
\newcommand{\MostPermissiveNetwork}{ \Network{M} }
\newcommand{\CuttableNetwork}{ \Network{C} }
\newcommand{\IntervalNetwork}{ \Network{I} }
\newcommand{\TrappingNetwork}{ \Network{T} }
\newcommand{\SubcubeNetwork}{ \Network{S} }

\title{Bringing memory to Boolean networks: a unifying framework}

\author{Maximilien Gadouleau$^1$, Loïc Paulevé$^2$, Sara Riva$^3$}
\date{$^1$ Department of Computer Science, Durham University, Durham, UK\\
    $^2$ Univ. Bordeaux, CNRS, Bordeaux INP, LaBRI, UMR 5800, F-33400 Talence, France\\
$^3$ Univ. Lille, CNRS, Centrale Lille, UMR 9189 CRIStAL, F-59000 Lille, France}

\begin{document}

\maketitle

\textbf{Abstract}.
Boolean networks are extensively applied as models of complex dynamical systems,
aiming at capturing essential features related to causality and synchronicity of the state changes of
components along time. Dynamics of Boolean networks result from the application of their Boolean
map according to a so-called update mode, specifying the possible transitions between network configurations. In this paper, we explore update modes that possess a memory on past configurations,
and provide a generic framework to define them. We show that recently introduced modes such as
the most permissive and interval modes can be naturally expressed in this framework, and we propose
novel update modes, the history-based, trapping, and subcube-based modes. Building on the unified definitions, we provide a comprehensive comparison of memory-based update modes, resulting in
their hierarchy by simulation and weak simulation. Finally, we highlight consequences of introducing
memory on the notions of trajectory and attractors.

\section{Introduction} \label{section:introduction}

Boolean networks (BNs) are a fundamental framework for addressing complex systems, with prominent
applications in biology, ecology, and social sciences~\cite{Montagud22,Rozum_2024,Gaucherel_2017,Poindron_2021,Grabisch_2013}.
Closely related to cellular automata, BNs consider a finite number of components, or automata,
each having its own neighborhood and rule for computing its next Boolean state from the configuration of the network.

A large amount of theoretical work underlined the importance of the scheduling
of the update of states of components for the dynamics that can be generated, the so-called \emph{update mode}.
One natural mode is the parallel (or synchronous) mode, where all components are
updated simultaneously and instantaneously, generating deterministic dynamics.
Other classical update modes are the (fully) asynchronous mode
where only one component can be updated at a time,
and the general asynchronous mode
where any subset of components can be updated simultaneously, potentially leading to non-deterministic dynamics.
By definition, the general asynchronous trajectories subsume the asynchronous and synchronous
trajectories.
The converse is, in general, false.
Update modes also impact the so-called \emph{limit} configurations, that are the configurations from which
any reachable configuration can return to them.
Limit configurations constitute the \emph{attractors} of dynamics as they represent a set of configurations that, once reached, can no longer be evaded. Limit configurations are a prime subject of study
in the literature, notably for determining their number and size, e.g., \cite{KLemm2005,Demongeot2011,Dubrova2011,PR11-SASB}, as well as their
robustness to update modes, e.g., \cite{AGMS09,ADFM13}.

When employing BNs as models of complex dynamical systems, such as gene regulation
within cells, update modes aim at
reflecting hypotheses related to the duration, speed, probability, and other quantitative features of
the transitions.
For instance, asynchronous modes can be motivated by the fact that genes can require different
amounts of proteins to become active~\cite{THOMAS1973563,TA90}.

There is a rich catalog of update modes defined in the literature that consider
various restrictions related to the synchronicity and sequentiality of updates.
These modes can be compared in terms of weak simulation relations:
an update mode $\mu$ weakly simulates another update mode $\mu'$ if, for any BN, for
any pair of its configurations, if there is a trajectory between these configurations with $\mu'$,
then there is a trajectory with $\mu$ as well.
This relationship enables to draw a hierarchy of update modes~\cite{PS22}, where most of them are
weakly simulated by the general asynchronous mode.
There is a notable exception with a couple of recently introduced unconventional update modes that enable generating trajectories
impossible with the general asynchronous mode.
The \emph{interval} update mode~\cite{interval}, inspired by work on concurrent systems by the means of Petri
nets, considers cases when the state change of a component can be put on hold for a while.
During that period, other components still observe the old value of it, and can evolve accordingly
and before the new state is finally revealed.
This can generate new trajectories, leading to configurations that are not
reachable with the general asynchronous mode.
The \emph{most permissive} update mode~\cite{pauleve2020reconciling} generalizes this abstraction by considering that a
component in the process of changing of state can be viewed as in superposition of both states.
This mode was motivated by the Boolean modeling of quantitative systems:
e.g., during the increase of the state of a component, there can be a time when it is sufficiently
high for acting on one component and not high enough for acting on another one.
Lately, the \emph{cuttable extension} update mode~\cite{cuttable} was introduced as a mode generating more
trajectories than the interval, but less than the most permissive, in order to better capture
monotone changes of component states.

One of the common ingredients between the interval, cuttable, and most permissive
modes  is the account of some kind of \emph{memory} of the state changes: it is
because they remember that a component is changing that they
can interleave additional transitions before applying the change, generating additional but plausible trajectories.
Note that another mode named \emph{BNs with memory} has been recently
introduced~\cite{J-Goles2020}.
However, as demonstrated in~\cite{automata21}, it actually boils down to having a fixed set of components
to update asynchronously while another set updates in parallel.

In this paper, we aim at emphasizing on the effect of memory in update modes of BNs, both by introducing new modes around this feature, and by studying their relationship with others, as well as related combinatorial properties. 

To that end, we propose a unifying framework for defining update modes with memory by characterizing the trajectories they generate. This allows to account for a sequence of configurations that have been computed since an initial one, and compute elongation of the trajectory based on this sequence.
Besides reformulating above-mentioned update modes, we take advantage of the framework to introduce novel elementary update modes using memory:
the \emph{history-based}, the \emph{trapping}, and the \emph{subcube-based} update modes.
Intuitively, the history-based mode extends the asynchronous by allowing to take any configuration of the past, instead of the latter only, to compute the next state of the component.
Then, the trapping mode extends the history-based by allowing to take any configuration of the past as the base next configuration.
Informally, this could be seen as some sort of time traveling: after the update of
some components, one can return to a past configuration and enable state changes based on knowledge of the future (which can then be different).
Finally, the subcube-based mode relies on the smallest subcube enclosing the configurations of the trajectory: the next configuration results from the update of any configuration of that subcube, with the state of one component computed from to any other configuration of the subcube.

One can naturally derive that the asynchronous trajectories are a subset of history-based trajectories, being themselves a subset of trapping trajectories.
In this paper, we establish simulation and weak simulation relationships with the interval,
cuttable, most permissive and introduced update modes (Sect.~\ref{section:comparing_update_modes}; Fig.~\ref{figure:hierarchy}).
In particular, we demonstrate that cuttable and history-based modes are incomparable, while both weakly simulating the interval, and both being weakly simulated by the most permissive, which in turns, is weakly simulated by the trapping mode.
We also show that the trapping mode is weakly bisimilar to the subcube-based mode.
We provide a similar hierarchy by considering inclusion of trajectories (simulation).

Another important result is that most permissive and trapping modes have the same limit configurations, and have the same limit configurations reachable from any configuration (Sect.~\ref{section:consequences}).

The paper is structured as follows. Sect.~\ref{section:update_modes} introduces our generic framework to express memory-based update modes which is employed to define novel trapping, history-based, and subcube-based update modes, and give equivalent definitions of asynchronous, most permissive, interval and cuttable update modes.
Then, Sect.~\ref{section:comparing_update_modes} provides a simulation and weak-simulation hierarchy
between those updates,
and Sect.~\ref{section:consequences} highlight important consequences of using memory in update
modes on trajectories and attractors.
Finally, Sect.~\ref{section:discussion} discusses the contributions of the paper.

\subsection*{Notations} \label{subsection:notations}

We denote the Boolean set by $\B = \{0,1\}$ and for any positive integer $n$, we denote $[n] = \{1, \dots, n \}$. A \Define{configuration} is $x = (x_1, \dots, x_n) \in \B^n$. For any $S \subseteq [n]$, we denote $x_S = (x_s : s \in S)$ and $x_{-S} = ( x_t : t \notin S )$, and we use the notation $x = (x_S, x_{-S})$. We shall identify an element $i \in [n]$ with the corresponding singleton $\{ i \}$, so that $x = (x_i, x_{-i})$ for instance. For any two configurations $x, y \in \B^n$, we denote the set of components where they differ by $\Delta( x, y ) = \{ i \in [n] : x_i \ne y_i \}$ and their Hamming distance by $\HammingDistance( x, y ) = | \Delta(x, y) |$. For any Boolean variable $a \in \B$, we denote its negation by $\neg a = 1 - a$; we extend this notation to configurations of any length by componentwise negation: $\neg x = ( \neg x_1, \dots, \neg x_n )$.

A \Define{subcube} of $\B^n$ is any $X \subseteq \B^n$ such that there exist two disjoint sets of configurations $S, T \subseteq [n]$ with $X = \{ x \in \B^n : x_S = 0, x_T = 1 \}$. For any set $A \subseteq \B^n$, the \Define{principal subcube} of $A$, denoted by $[A]$, is the smallest subcube containing $A$. If $A = \{a_1, \dots, a_k\}$, we also denote $[A] = [a_1, \dots, a_k]$. If $X$ is a subcube and $x \in X$, then there is a unique $y \in X$ such that $X = [x, y]$; we refer to $y$ as the \Define{opposite} of $x$ in $X$, and we denote it by $y = X - x$. 

A \Define{Boolean network} (BN) with $n$ components is a mapping $f : \B^n \to \B^n$. We denote the set
of BNs with $n$ components as $\Functions(n)$.
Any BN $f \in \Functions(n)$ can be viewed as $f = (f_1, \dots, f_n)$ where $f_i : \B^n \to \B$ is given by $f_i(x) = f(x)_i$ for all $i \in [n]$. 

A \Define{trapspace} of $f \in \Functions(n)$ is a subcube $X \subseteq \B^n$ such that $f(X) \subseteq X$. 
The collection of all trapspaces of $f$ is closed under intersection. 
Then for any $x \in \B^n$, there is a smallest trapspace of $f$ that contains $x$, which we shall refer to as the \Define{principal trapspace} of $x$ (with respect to $f$). For the sake of simplicity, we denote it by $T_f(x)$. The principal trapspace $T_f(x)$ can be recursively computed as follows: let $T_0 = \{ x \}$ and $T_i = [T_{i-1} \cup f( T_{i-1} )]$, then $T_n = T_f(x)$. In particular, $[x, f(x)] \subseteq T_f(x)$. 
A trapspace $T$ is \Define{minimal} if there is no trapspace $T'$ with $T' \subsetneq T$. Clearly, any minimal trapspace is principal, but the converse does not necessarily hold. 
%

A (directed) \Define{graph} is $\Gamma  = (V, E)$, where $V$ is the set of vertices and $E \subseteq V^2$ is the set of edges. 
The \Define{asynchronous graph} of a BN $f \in \Functions(n)$ is the graph $\Asynchronous(f) = (V, E)$ where $V = \B^n$ and 
\[
    E = \{ ( f_i( x ), x_{-i} ) : x \in \B^n, i \in [n] \}. 
\]
Note that in most literature, one removes the loops $(x, x)$ from the asynchronous graph, that occur every time $f_i(x) = x_i$, yet we shall keep those loops instead in our definition. However, when drawing the asynchronous graph, we shall not display the loops and instead draw the underlying hypercube with thin black lines and the arcs of the graph with thick blue arrows. See below for an example of a BN, for which we give the asynchronous graph and all the principal and minimal trapspaces.

\begin{example} \label{example:network}
Let $f \in \Functions(3)$ be defined as
\begin{center}
\begin{tabular}{c|c}
     $x$ & $f(x)$ \\
     \hline
     $000$ & $110$ \\ 
     $001$ & $100$ \\
     $010$ & $000$ \\
     $011$ & $110$ \\
     $100$ & $100$ \\
     $101$ & $101$ \\
     $110$ & $110$ \\
     $111$ & $110$ 
\end{tabular}
\end{center}

The asynchronous graph of $f$ is given by:

\begin{center}
\begin{tikzpicture}[scale=1.5]
    \node (000) at (0,0) {$000$};
    \node (001) at (1,1) {$001$};
    \node (010) at (0,2) {$010$};
    \node (011) at (1,3) {$011$};
    \node (100) at (2,0) {$100$};
    \node (101) at (3,1) {$101$};
    \node (110) at (2,2) {$110$};
    \node (111) at (3,3) {$111$};

    \path[draw] (000) -- (001) -- (011) -- (111)
    (000) -- (010) -- (110) -- (111)
    (000) -- (100) -- (101) -- (111)
    (001) -- (101)
    (010) -- (011)
    (100) -- (110);
    
    \draw[thick,-latex, blue] (000) -- (100);
    \draw[thick,-latex, blue] (000) -- (010);

    \draw[thick,-latex, blue] (001) -- (101);
    \draw[thick,-latex, blue] (001) -- (000);

    \draw[thick,-latex, blue] (010) -- (000);

    \draw[thick,-latex, blue] (011) -- (111);
    \draw[thick,-latex, blue] (011) -- (010);

    \draw[thick,-latex, blue] (111) -- (110);
\end{tikzpicture}
\end{center}

Then the principal trapspaces of $f$ are given below, where $T_f( 100 )$, $T_f( 101 )$ and $T_f( 110 )$ are the only minimal trapspaces of $f$:
\begin{alignat*}{3}  
    T_f(000) &= [000,110] &&= \{ x : x_3 = 0 \} \\
    T_f(001) &= [001,110] &&= \{ x \} =\mathbb{B}^3\\
    T_f(010) &= [010,100] &&= \{ x : x_3 = 0 \} \\
    T_f(011) &= [011,100] &&= \{ x \} =\mathbb{B}^3\\
    T_f(100) &= [100,100] &&= \{ 100 \} \\
    T_f(101) &= [101,101] &&= \{ 101 \} \\
    T_f(110) &= [110,110] &&= \{ 110 \} \\
    T_f(111) &= [111,110] &&= \{ x : x_{1}= x_{2} = 1 \}.
\end{alignat*}

\end{example}

\section{Update modes for dynamics with memory} \label{section:update_modes}

\subsection{A unified framework} \label{subsection:unified_framework}

In this section, we give a unified framework for defining update modes that rely on memory 
along trajectories. We focus on the fully asynchronous case, whereby one component updates its state at each given time.

Before giving the unified framework in its generality, we would like to illustrate it via the example of an asynchronous trajectory. Consider the network $f$ from Example \ref{example:network}. By following the arcs of the asynchronous graph, we have a trajectory 
\[
    x^0 = 011 \to x^1 = 010 \to x^2 = 000 \to x^3 = 100 \to x^4 = 100. 
\]
Each update is of the form 
\[
    x^{a-1} = ( x^{a-1}_{i^a}, x^{a-1}_{-i^a} ) \to x^a = ( f_{i^a}( x^{a-1} ), x_{-i^a}),
\]
where $i^a$ is the component updated at time $a$. We note that $i^4$ could be any index in $[n]$, since $x^3$ is a fixed point, but this is inconsequential. As such, we can succinctly represent each update by the index $i^a \in [n]$, and the whole trajectory can be obtained from its starting configuration $x^0 = 011$ and the sequence $i^1 = 3, i^2 = 2, i^3 = 1, i^4 = 1$. 

In order to add memory to our model of trajectories, we will allow two generalisations:
\begin{enumerate}
    \item the function $f_{i^a}$ may be applied to a \Define{source configuration} $s^a \in \B^n$, which may be distinct from $x^{a-1}$, so that $x^a_{i^a} = f_{i^a}( s^a )$;

    \item the rest of the configuration can be taken from a \Define{target configuration} $t^a \in \B^n$, which again may be distinct from $x^{a-1}$, so that $x^a_{-i^a} = t^a_{-i^a}$
\end{enumerate}
As such, the transition becomes
\begin{equation} \label{equation:update}
    x^{a-1} = ( x^{a-1}_{i^a}, x^{a-1}_{-i^a} ) \to x^a = ( f_{i^a}( s^a ), t^a_{-i^a}),
\end{equation}
and it can be represented by the triple
\[
    w^a = ( i^a, s^a, t^a ) \in [n] \times \B^n \times \B^n.
\]
Again, different choices for $w^a$ may lead to the same next configuration $x^a$, but an entire trajectory $x^0 \to \dots \to x^l$ can be obtained from $x^0$ and the sequence of triples $w^1, \dots, w^l$. For instance, the asynchronous trajectory above can be obtained from its starting configuration $x^0 = 011$ and the sequence of triples
\[
    w^1 = ( 3, 011, 011 ), \quad w^2 = ( 2, 010, 010 ), \quad w^3 = ( 1, 000, 000 ), \quad w^4 = ( 1, 100, 100 ).
\]

In its full generality, a (fully asynchronous) \Define{trajectory} is a sequence $(x^0, \dots, x^l) \in ( \B^n )^*$ such that for all $1 \le a \le l$, there exists $t \in [ x^0, \dots, x^{a-1} ]$ such that $\HammingDistance( t, x^a ) \le 1$. In particular, a \Define{geodesic} is a trajectory $(x^0, \dots, x^l)$ such that $\HammingDistance( x^{a-1}, x^a ) = 1$ for all $1 \le a \le l$ and $l = \HammingDistance(x^0, x^l)$; this is equivalent to a shortest path from $x^0$ to $x^l$ in the hypercube. A prefix of the trajectory $( x^0, \dots, x^l )$ is any $(x^0, \dots, x^a)$ for $0 \le a \le l$, or the empty sequence. An \Define{update mode} is a mapping $\mu$ which assigns, for any $n$ and any BN $f \in \Functions(n)$, a collection of trajectories closed by taking prefixes (i.e. if $(x^0, \dots, x^l) \in \mu( f )$, then $(x^0, \dots, x^a) \in \mu(f)$). Since we shall fix the BN in the remainder of this section, we shall omit the dependence on $f$ in our notation. To emphasize that $( x^0, \dots, x^l )$ is a $\mu$-trajectory, we write $x^0 \to_\mu \dots \to_\mu x^l$. A pair of configurations $(x,y)$ is a \Define{reachability pair} for an update mode $\mu$ if there exists a $\mu$-trajectory $x = x^0 \to_\mu \dots \to_\mu x^l = y$, and we say that $y$ is \Define{reachable} from $x$ and denote it with $x \to^*_\mu y$.

In this paper, we consider update modes that satisfy certain constraints. In particular, we assume that $\mu$ only produces trajectories where each update is of the form of Equation \eqref{equation:update} for some reasonable choices of source and target configurations. More precisely, the trajectories are $(x^0, \dots, x^l)$, where for each $1 \le a \le l$, there exists $w^a = ( i^a, s^a, t^a ) \in [n] \times \B^n \times \B^n$ such that:
\begin{enumerate}
    \item $s^a \in [x^0, \dots, x^{a-1}]$ is a configuration made up of states with different time stamps from the trajectory, so that $f_{i^a}$ is applied to state values that are drawn from the trajectory's history;

    \item similarly, $t^a \in [x^0, \dots, x^{a-1}]$, so that the remainder of the $x^a$ configuration is made up of state values drawn from the trajectory's history;

    \item the configuration $x^a$ is given by $x^a = ( f_{i^a}( s^a ), t^a_{-i^a})$.
\end{enumerate}

Different update modes correspond to different constraints over $i^a$, $s^a$, and $t^a$ (the case without any further constraints will be the so-called subcube update mode). We list below the update modes that we shall consider in this paper. None of those applies any restriction on $i^a$: it can be any component in $[n]$ at any time step.  We shall highlight some important properties of these update modes with memory in Section \ref{section:consequences}, once they have been presented and thoroughly compared.

\subsection{Update modes}\label{sec:updates}

\subsubsection{Asynchronous updates} 
    
We begin with the simplest update mode, namely asynchronous updates \cite{THOMAS1973563}. 
Let $f \in \Functions(n)$. The asynchronous updating mode considers $x \to y$ iff $\exists i \in [n]$ such that $y_i=f_i(x)$ and $y_{-i} = x_{-i}$.

 In this case, a component is updated at each step and the result is calculated and applied to the last configuration visited.
Then, we select $i^a$ and apply $f_{i^a}$ to the current configuration $x^{a-1}$. In other words, $x \to_{\Asynchronous} y$ if and only if $(x,y)$ is an arc in the asynchronous graph of $f$.
Thus, the memory consists only of the current configuration.
In our framework, the definition is as follows.

\paragraph{Asynchronous updates $\Asynchronous$:} All trajectories $(x^0, \dots, x^l)$ obtained by $x^0 \in \B^n$ and $w^a = ( i^a, s^a, t^a ) \in [n] \times \B^n \times \B^n$ for $1 \le a \le l$, where
\begin{itemize}
    \item source $s^a = x^{a-1}$;

    \item target $t^a = x^{a-1}$.

\end{itemize}

\begin{example} \label{example:asynchronous_trajectory}
Let us consider the network $\AsynchronousNetwork \in \Functions(2)$ defined by
\begin{align*}
    \AsynchronousNetwork_1( x ) &= \neg x_1 \lor \neg x_2, \\
    \AsynchronousNetwork_2( x ) &= \neg x_1 \lor \neg x_2.
\end{align*}

From the asynchronous graph of $\AsynchronousNetwork$ (shown below), we can see that $00 \to_\Asynchronous^* 11$. A possible asynchronous trajectory $00 \to_\Asynchronous 10 \to_\Asynchronous 11 \to_\Asynchronous 01$ is detailed in the table below.

\begin{minipage}[]{.5\linewidth}
    \centering
    \begin{tikzpicture}
        \node (00) at (0,0) {$00$};
        \node (01) at (0,2) {$01$};
        \node (10) at (2,0) {$10$};
        \node (11) at (2,2) {$11$};
    
        \path[draw] (00) -- (01) -- (11)
        (00) -- (10) -- (11);
        
        \draw[thick,-latex, blue] (00) -- (10);
        \draw[thick,latex-latex, blue] (01) -- (11);
        \draw[thick,-latex, blue] (00) -- (01);
        \draw[thick,latex-latex, blue] (10) -- (11);
    \end{tikzpicture}
\end{minipage}%
\begin{minipage}[]{.3\linewidth}
    \centering
    \begin{tabular}{c|c|c|c|c}
         $a$    & $i^a$ & $s^a$ & $t^a$ & $x^a$ \\
         \hline
         $0$    & & & & $00$ \\
         $1$    & $1$   & $00$  & $00$  & $10$ \\
         $2$    & $2$   & $10$  & $10$  & $11$ \\
         $3$    & $1$   & $11$  & $11$  & $01$ \\
    \end{tabular}
\end{minipage}

\end{example}

\subsubsection{History-based updates}

We now inject some amount of memory in our update mode. We introduce the history-based updates, where the component $i^a$ has a delay in getting the information about the current configuration and in fact works on a superseded version thereof. As a result, $f_{i^a}$ is not necessarily applied to the current configuration $x^{a-1}$, but to any configuration $s^a \in \{x^0, \dots, x^{a-1} \}$ that has already occurred. In our framework, this becomes:

\paragraph{History-based updates $\History$:} All trajectories $(x^0, \dots, x^l)$ obtained by $x^0 \in \B^n$ and $w^a = ( i^a, s^a, t^a ) \in [n] \times \B^n \times \B^n$ for $1 \le a \le l$, where
\begin{itemize}
    \item source $s^a \in \{x^0, \ldots, x^{a-1}\}$;

    \item target $t^a = x^{a-1}$.

\end{itemize}

\begin{example} \label{example:history_trajectory}
Let us consider the network $\HistoryNetwork \in \Functions(3)$ given by
\begin{align*}
    \HistoryNetwork_1( x )    &= 1, \\
    \HistoryNetwork_2( x )    &= x_1, \\
    \HistoryNetwork_3( x )    &=  x_2.
\end{align*}

The asynchronous graph of $\HistoryNetwork$ is given below, as well as a possible $\History$-trajectory. The table is a possible representation of the trajectory $000 \to_\History 100 \to_\History 110 \to_\History 111 \to_\History 101$. It shows how $i^a$, $s^a$ and $t^a$ are chosen in each step. A new equivalent representation is proposed in Figure \ref{figure:trajectories1}.

\begin{minipage}[]{.5\linewidth}
    \centering
    \begin{tikzpicture}
        \node (000) at (0,0) {$000$};
        \node (001) at (1,1) {$001$};
        \node (010) at (0,2) {$010$};
        \node (011) at (1,3) {$011$};
        \node (100) at (2,0) {$100$};
        \node (101) at (3,1) {$101$};
        \node (110) at (2,2) {$110$};
        \node (111) at (3,3) {$111$};
    
        \path[draw] (000) -- (001) -- (011) -- (111)
        (000) -- (010) -- (110) -- (111)
        (000) -- (100) -- (101) -- (111)
        (001) -- (101)
        (010) -- (011)
        (100) -- (110);
        
        \draw[thick,-latex, blue] (000) -- (100);
        \draw[thick,-latex, blue] (001) -- (101);
        \draw[thick,-latex, blue] (010) -- (110);
        \draw[thick,-latex, blue] (011) -- (111);
        \draw[thick,-latex, blue] (010) -- (000);
        \draw[thick,-latex, blue] (011) -- (001);
        \draw[thick,-latex, blue] (100) -- (110);
        \draw[thick,-latex, blue] (101) -- (111);
        \draw[thick,-latex, blue] (001) -- (000);
        \draw[thick,-latex, blue] (101) -- (100);
        \draw[thick,-latex, blue] (010) -- (011);
        \draw[thick,-latex, blue] (110) -- (111);
    \end{tikzpicture}
\end{minipage}%
\begin{minipage}[]{.3\linewidth}
    \centering
    \begin{tabular}{c|c|c|c|c}
         $a$    & $i^a$ & $s^a$ & $t^a$ & $x^a$ \\
         \hline
         $0$ & & & & $000$ \\
         $1$ & $1$ & $000$ & $000$ & $100$ \\
         $2$ & $2$ & $100$ & $100$ & $110$ \\
         $3$ & $3$ & $110$ & $110$ & $111$ \\
         $4$ & $2$ & $000$ & $111$ & $101$ 
    \end{tabular}
\end{minipage}

\end{example}

\subsubsection{Trapping updates} 
So far, all our updates applied $f_{i^a}$ to some source configuration $s^a$ and copied the rest of the current configuration: $x^a_{-i^a} = x^{a-1}_{ -i^a }$. Let us now consider an update mode in which the rest of the configuration can be taken from the history of the trajectory. Trapping updates are the natural generalisation of history-based updates, where both the source $s^a$ and the target $t^a$ are chosen from the set $\{ x^0, \dots, x^{a-1} \}$.
Thus, it allows accounting for possible delays in updating a component, and it allows the result of the
update to be applied on a past configuration, becoming the new current configuration.
The term ``Trapping updates'' come from the fact, proved in the sequel, that $x \to_\TrappingGraph^* y$ if and only if $y \in T_f(x)$.
In our framework, this turns out to be formalised as follows.

\paragraph{Trapping updates $\TrappingGraph$:} All trajectories $(x^0, \dots, x^l)$ obtained by $x^0 \in \B^n$ and $w^a = ( i^a, s^a, t^a ) \in [n] \times \B^n \times \B^n$ for $1 \le a \le l$, where
\begin{itemize}
    \item source $s^a \in \{ x^0, \dots, x^{a-1} \}$;

    \item target $t^a \in \{ x^0, \dots, x^{a-1} \}$.

\end{itemize}

\begin{example} \label{example:trapping_trajectory}
Let us consider the network $\TrappingNetwork \in \Functions(2)$ defined by 
\begin{align*}
    \TrappingNetwork_1(x) &= 1, \\
    \TrappingNetwork_2(x) &= x_1 \lor x_2.
\end{align*}

The asynchronous graph of $\TrappingNetwork$ is given below, as well as a possible $\TrappingGraph$-trajectory showing how $00 \to^*_{\TrappingGraph} 01$. Note how $x^3$ is computed with $t^3 = 00 \ne x^2$.

\begin{minipage}[]{.5\linewidth}
    \centering
    \begin{tikzpicture}
        \node (00) at (0,0) {$00$};
        \node (01) at (0,2) {$01$};
        \node (10) at (2,0) {$10$};
        \node (11) at (2,2) {$11$};
    
        \path[draw] (00) -- (01) -- (11)
        (00) -- (10) -- (11);
        
        \draw[thick,-latex, blue] (00) -- (10);
        \draw[thick,-latex, blue] (01) -- (11);
        \draw[thick,-latex, blue] (10) -- (11);
    \end{tikzpicture}
\end{minipage}%
\begin{minipage}[]{.3\linewidth}
    \centering
    \begin{tabular}{c|c|c|c|c}
         $a$    & $i^a$ & $s^a$ & $t^a$ & $x^a$ \\
         \hline
         $0$    & & & & $00$ \\
         $1$    & $1$   & $00$  & $00$  & $10$ \\
         $2$    & $2$   & $10$  & $10$  & $11$ \\
         $3$    & $2$   & $11$  & $00$  & $01$ \\
    \end{tabular}
\end{minipage}

\end{example}

\subsubsection{Most permissive updates}\label{sec:mp}

We now focus on the most permissive updates introduced in \cite{pauleve2020reconciling}. This updating technique can be defined using dynamic states or hypercubes, let us recall the definition in this first form.


In this update mode, the change of state of the components is broken down into two stages, with a dynamic state indicating an increase $\nearrow$ or decrease $\searrow$. The set of possible states of a component is denoted $\mathbb{P} = \{0,\nearrow , \searrow, 1\}$. When a component is in a dynamic state, the other components in the network can arbitrarily read $0$ or $1$.
The function $\gamma : \mathbb{P}^n  \to 2^{\B^n}$ associates with any configuration $x \in \mathbb{P}^n$ the set of configurations that correspond to it. Then, $\gamma(x)=\{x'\in \B^n \mid \forall i \in [n], x_i \in \B \Rightarrow x'_i =x_i\}$. 
In this semantics, it is therefore defined that, $\forall x,y \in \mathbb{P}^n$, $x\to y$ iff $\exists i\in [n]:\Delta(x,y)=\{i\}$ and 
$$y_i=\begin{cases}
    \nearrow &\text{if }x_i<f_i(z)\text{ for some } z \in \gamma(x)\\
    1 & \text{if }x_i=\nearrow \\
    \searrow &\text{if }x_i>f_i(z)\text{ for some } z \in \gamma(x)\\
    0 & \text{if }x_i=\searrow. \\
\end{cases}
$$

The most permissive update mode considers that a component in the process of changing of state can be viewed as in superposition of both states. As such, in our framework, this is equivalent to allowing a different delay for any $j$, hence $f_{i^a}$ is now applied to some $s^a$ where each $s^a_j$ can be chosen from $\{ x^0_j, \dots, x^{a-1}_j \}$; this is equivalent to $s^a \in [ x^0, \dots, x^{a-1} ]$. 


The most permissive update mode can therefore be formalised in our framework as follows.

\paragraph{Most permissive updates $\MostPermissive$:} All trajectories $(x^0, \dots, x^l)$ obtained by $x^0 \in \B^n$ and $w^a = ( i^a, s^a, t^a ) \in [n] \times \B^n \times \B^n$ for $1 \le a \le l$, where
\begin{itemize}
    \item source $s^a \in [ x^0, \dots, x^{a-1} ]$;

    \item target $t^a = x^{a-1}$.

\end{itemize}

The correspondence between the two definitions is given in Sect.~\ref{seq:eq-mp}.

In \cite{pauleve2020reconciling}, most permissive updates were motivated by the abstraction of quantitative systems by
Boolean networks: the fact that the source configuration can be chosen within the subcube
$[x^0,\dots,x^{a-1}]$ aims at capturing the potential heterogeneity of influence thresholds and time
scales between component updates.
They demonstrated that most permissive updates allow capturing all the dynamics possible in any quantitative refinement
of the BN, while the (general) asynchronous mode hinder the prediction of several observed trajectories.

\begin{example} \label{example:most_permissive_trajectory}
Let us consider the network $\MostPermissiveNetwork \in \Functions(3)$ defined by 
\begin{align*}
    \MostPermissiveNetwork( 000 ) &= 110, \\
    \MostPermissiveNetwork( 010 ) &= 011, \\
    \MostPermissiveNetwork( x ) &= x \text{ otherwise.}
\end{align*}

The asynchronous graph of $\MostPermissiveNetwork$  is given below, as well as a possible $\MostPermissive$-trajectory showing how $000 \to^*_{\MostPermissive} 111$. 

\begin{minipage}[]{.5\linewidth}
    \centering
    \begin{tikzpicture}
        \node (000) at (0,0) {$000$};
        \node (001) at (1,1) {$001$};
        \node (010) at (0,2) {$010$};
        \node (011) at (1,3) {$011$};
        \node (100) at (2,0) {$100$};
        \node (101) at (3,1) {$101$};
        \node (110) at (2,2) {$110$};
        \node (111) at (3,3) {$111$};
    
        \path[draw] (000) -- (001) -- (011) -- (111)
        (000) -- (010) -- (110) -- (111)
        (000) -- (100) -- (101) -- (111)
        (001) -- (101)
        (010) -- (011)
        (100) -- (110);
        
        \draw[thick,-latex, blue] (000) -- (100);
        \draw[thick,-latex, blue] (000) -- (010);
        \draw[thick,-latex, blue] (010) -- (011);
    \end{tikzpicture}
\end{minipage}%
\begin{minipage}[]{.3\linewidth}
    \centering
    \begin{tabular}{c|c|c|c|c}
         $a$    & $i^a$ & $s^a$ & $t^a$ & $x^a$ \\
         \hline
         $0$ & & & & $000$ \\
         $1$ & $1$ & $000$ & $000$ & $100$ \\
         $2$ & $2$ & $000$ & $100$ & $110$ \\
         $3$ & $3$ & $010$ & $110$ & $111$ 
    \end{tabular}
\end{minipage}

\end{example}

\subsubsection{Subcube-based updates}

Subcube-based updates are defined similarly to trapping updates, but instead consider the subcube $[ x^0, \dots, x^{a-1} ]$
containing the configuration already visited for both the source $s^a$ and the target $t^a$.

\paragraph{Subcube-based updates $\Subcube$:} All trajectories $(x^0, \dots, x^l)$ obtained by $x^0 \in \B^n$ and $w^a = ( i^a, s^a, t^a ) \in [n] \times \B^n \times \B^n$ for $1 \le a \le l$, where
\begin{itemize}
    \item source $s^a \in [ x^0, \dots, x^{a-1} ]$;

    \item target $t^a \in [ x^0, \dots, x^{a-1} ]$.

\end{itemize}

\begin{example} \label{example:subcube_trajectory}
Let us consider the network $\SubcubeNetwork \in \Functions(3)$ given by
\begin{align*}
    \SubcubeNetwork(000) &= 100 \\
    \SubcubeNetwork(100) &= 110 \\
    \SubcubeNetwork(110) &= 111 \\
    \SubcubeNetwork(x) &= x \text{ otherwise}.
\end{align*}
The asynchronous graph of $\SubcubeNetwork$ is given below, as well as a possible $\Subcube$-trajectory showing how $000 \to^*_{\Subcube} 001$. 

\begin{minipage}[]{.5\linewidth}
    \centering
        \begin{tikzpicture}
            \node (000) at (0,0) {$000$};
            \node (001) at (1,1) {$001$};
            \node (010) at (0,2) {$010$};
            \node (011) at (1,3) {$011$};
            \node (100) at (2,0) {$100$};
            \node (101) at (3,1) {$101$};
            \node (110) at (2,2) {$110$};
            \node (111) at (3,3) {$111$};
        
            \path[draw] (000) -- (001) -- (011) -- (111)
            (000) -- (010) -- (110) -- (111)
            (000) -- (100) -- (101) -- (111)
            (001) -- (101)
            (010) -- (011)
            (100) -- (110);
            
            \draw[thick,-latex, blue] (000) -- (100);
            \draw[thick,-latex, blue] (100) -- (110);
            \draw[thick,-latex, blue] (110) -- (111);
        
        \end{tikzpicture}
\end{minipage}%
\begin{minipage}[]{.3\linewidth}
    \centering        
        \begin{tabular}{c|c|c|c|c}
                 $a$    & $i^a$ & $s^a$ & $t^a$ & $x^a$ \\
                 \hline
                 $0$ & & & & $000$ \\
                 $1$ & $1$ & $000$ & $000$ & $100$ \\
                 $2$ & $2$ & $100$ & $000$ & $010$ \\
                 $3$ & $3$ & $110$ & $000$ & $001$ 
        \end{tabular}
\end{minipage}
\end{example}

\subsubsection{Interval updates}\label{sec:interval}

The interval update was initially introduced for Petri nets in \cite{interval-pn}, and later adapted to
BNs in \cite{interval}, as a mean to capture transitions that could occur while a component is
changing of state, i.e., while the component is committed to change, but the change has not been
applied yet. This was formalised by proposing to consider that each component $i \in [n]$ is decoupled in two nodes: a \emph{write} node $(2i - 1)$ storing the next value and a \emph{read} node $(2i)$ for the current value.
Whenever a state change for a component is triggered, only its write node is updated first.
Then, the read node will be updated later, but in the meantime, other components can change of value
based on the read node state.

Let us define a function $\tau : \B^{2n} \to \B^n$ which map a configuration of the interval semantics to a configuration of the BN $f$ by projecting on the read nodes, i.e., $\tau(z)_i = z_{2i}$ for every $i \in [n]$.
Given a BN $f$ with $n$ components, $\Tilde{f}$ is a BN with $2n$ components where $\forall i \in [n]$, we have 
$\Tilde{f}_{2i}(z)=z_{2i-1}$ and
$\Tilde{f}_{2i-1}(z)= \begin{cases}f_i(\tau(z)) &\text{if }
z_{2i} = z_{2i-1}\\ z_{2i-1}&\text{otherwise}\end{cases}$.

\begin{example}\label{ex:intervalclassic}
    Let us consider the network $\IntervalNetwork \in \Functions(3)$ given by
\begin{align*}
    \IntervalNetwork(000) &= 111, \\
    \IntervalNetwork(100) &= 101, \\
    \IntervalNetwork(101) &= 111, \\
    \IntervalNetwork(110) &= 010, \\
    \IntervalNetwork(x) &= x \text{ otherwise}.
\end{align*}
We can obtain the following possible sequence of asynchronous iterations of $\Tilde{f}$ with the interval semantics. We insert spaces in the configurations to facilitate the identification of the two read/write components corresponding to each of the original components.
$$00\; 00\; 00 \to 10\;00\;00 \to 10\;00\;10 \to 10\; 10\; 10\to 11\;10\;10 \to 11\;11\;10\to 01\;11\;10\to 00\;11\;10\to 00\;11\;11$$
Hence, the configuration $011$ of $\IntervalNetwork$ is reachable from $000$ according to this update mode.
\end{example}

This decoupling introduces a natural notion of memory: while the update of the component is pending,
its previous state is kept.
Compared to the history-based update mode, this forbids to update a component with respect to any
previous configuration, but only from the most recent state of components, ignoring any pending
change.
Indeed, consider the following scenario for history-based updates: the same component is updated twice, say $a < b$ and $i^a = i^b = i$, with $s^a = x^{ a' }$ and $s^b = x^{ b' }$ such that $a' > b'$. In that scenario, the first update $a$ uses a more recent version $x^{ a' }$ of the configuration that the second update $b$.
Such a scenario cannot occur in interval updates.

In our framework, the definition of the interval update requires coupling each configuration $x^a$ of a trajectory with a vector $V^a$ associating to each component the time to read its state.
Intuitively, for any component $j$, having $V^a_{j} < a-1$ enables to use its state prior to its
latest state change (thus, $j$ state change is pending), whereas $V^a_j = a-1$ ensures that its
state change has been completed, and other components read it.
When a component $i$ is updated at time $a$, we impose that $V^a_i=a-1$, i.e., the most recent state
of $i$ must be used. This captures the fact that a component has to have completed its state change
before being updated again, which is enforced by the original definition of the Interval update
mode~\cite{interval-pn,interval}.
We also enforce $V^a \ge V^{a-1}$ (i.e., for any component $j$, $V^a_j \ge V^{a-1}_j$), which forbids the scenario described above for history-based updates.

\paragraph{Interval updates $\Interval$:} All trajectories $(x^0, \dots, x^l)$ obtained by $x^0 \in \B^n$ and $w^a = ( i^a, s^a, t^a ) \in [n] \times \B^n \times \B^n$ for $1 \le a \le l$, where
\begin{itemize}
    \item vector $V^a \in \N^n$ such that $V^0 = 0^n$, $V^a \geq V^{a-1}$, $V^a \leq (a-1)^n$ and $V^a_{i^a} = a-1$;

    \item source $s^a$ with $s^a_j = x^{ V^a_j }_j$ for all $j \in [n]$;

    \item target $t^a = x^{a-1}$.

\end{itemize}

Notice that, in general, several $V^a$ can be possible, modelling different schedule of state change
completions, possibly enabling reaching different configurations.
The asynchronous update mode, where no memory of previous states is used, corresponds to the case
whenever $V^a=(a-1)^n$ for any time $a \geq 1$.
The correspondence between the two definitions is given in Sect.~\ref{sec:eq-interval}.
Intuitively, the memory-based trajectories correspond to the projection of sequences of interval
iterations over the \emph{write} nodes and $V^a$ reflect the \emph{read} nodes.

\begin{example} \label{example:interval_trajectory}
Let us consider the network of the previous example. The asynchronous graph of $\IntervalNetwork$ is shown below. We also show an example of $\Interval$-trajectory.

\medskip

\begin{minipage}[]{.5\linewidth}
    \centering
    \begin{tikzpicture}
        \node (000) at (0,0) {$000$};
        \node (001) at (1,1) {$001$};
        \node (010) at (0,2) {$010$};
        \node (011) at (1,3) {$011$};
        \node (100) at (2,0) {$100$};
        \node (101) at (3,1) {$101$};
        \node (110) at (2,2) {$110$};
        \node (111) at (3,3) {$111$};
    
        \path[draw] (000) -- (001) -- (011) -- (111)
        (000) -- (010) -- (110) -- (111)
        (000) -- (100) -- (101) -- (111)
        (001) -- (101)
        (010) -- (011)
        (100) -- (110);
        
        \draw[thick,-latex, blue] (000) -- (001);
        \draw[thick,-latex, blue] (000) -- (010);
        \draw[thick,-latex, blue] (000) -- (100);
        \draw[thick,-latex, blue] (100) -- (101);
        \draw[thick,-latex, blue] (101) -- (111);
        \draw[thick,-latex, blue] (110) -- (010);
    
    \end{tikzpicture}
\end{minipage}%
\begin{minipage}[]{.3\linewidth}
    \centering   
    \begin{tabular}{c|c|c|c|c|c}
    $a$ & $i^a$ & $V^a$ & $s^a$ & $t^a$ & $x^a$ \\
    \hline
    $0$ &       & $\begin{bmatrix} 0 & 0 & 0 \end{bmatrix}$ &       &       & $000$ \\
    $1$ & $1$   & $\begin{bmatrix} 0 & 0 & 0 \end{bmatrix}$ & $000$ & $000$ & $100$ \\
    $2$ & $3$   & $\begin{bmatrix} 0 & 0 & 1 \end{bmatrix}$ & $000$ & $100$ & $101$ \\
    $3$ & $2$   & $\begin{bmatrix} 0 & 2 & 1 \end{bmatrix}$ & $000$ & $101$ & $111$ \\
    $4$ & $1$   & $\begin{bmatrix} 3 & 3 & 1 \end{bmatrix}$ & $110$ & $111$ & $011$
    \end{tabular}
\end{minipage}

\medskip
In the trajectory shown, all components are first updated once with using the memory of the initial
state (as done in the previous example). Then, the first component is updated again based on the most recent value of components 1, 2,
but using the previous state of component 3 ($V^4_3=1$, i.e., the time just before the last update
of $3$).
Indeed, with $V^4 = \begin{bmatrix} 3 & 3 & 1 \end{bmatrix}$, we obtain $s^4=x_1^3x_2^3x_3^1=110$ and $\IntervalNetwork(110)_1 = 0$.
Thus, $x^4=011$. It can be seen that the trajectories in this example and the previous one are equivalent. Intuitively, updating the read component based on the write component corresponds to updating the value in $V^a$.
\end{example}

\subsubsection{Cuttable updates}\label{sec:cuttable}

Finally, we consider the cuttable update mode introduced in \cite{cuttable}, which generalizes the
interval update mode by considering a finer description of the delay of state changes.
Essentially, with the interval update mode, a component can have a delay to apply (transmit) its state change.
Once it is transmitted (the read node is updated), all the components of the network will read its new state.
The cuttable update mode allows having different delays of transmission for each component of the
network.
Thus, instead of having a read node per component, there is a read node per pair of components. The
read node copies the state of the (unique) write node of components asynchronously.
Hence, after a change of state of a component $i$ has been fired, there can be moments whenever
a component $j$ has access to its newer state, while another component $k$ has access to its
previous state.

\smallskip

Let us recall the definition of an $E$-extension of a Boolean network $f$. For simplicity, let us
define $E$ as the set of pairs $(i,j)$ such that component $i$ influences the value of component $j$
(i.e., $x_i$ is involved in the function $f_j$). It is possible to define the BN $\Tilde f:\B^{n+|E|} \to \B^{n+|E|}$ such that
\begin{alignat*}{3}
    \Tilde f_i(x) &=  f_i(\tau^i(x)) &&\quad \text{ for all } \ i\in[n],\\
    \Tilde f_e(x) &= x_j  &&\quad \text{ if } e=(j,k)\in E,
\end{alignat*}
    
where $\tau^i:\B^{n+|E|} \to \B^n$ is defined, for all $j\in [n]$ as:
$$\tau^i(x)_j=\begin{cases}
    x_{(j,i)} &\text{ if }(j,i)\in E\\
    x_j &\text{ otherwise. }
    \end{cases}$$

\begin{example}
     Let us consider the network $\CuttableNetwork \in \Functions(3)$ given by
    \begin{align*}
        \CuttableNetwork_1( x )    &= 1, \\
        \CuttableNetwork_2( x )    &= x_1, \\
        \CuttableNetwork_3( x )    &=  x_2 \land \neg x_1.
    \end{align*}
    The set $E$ is $\{(1,2),(2,3),(1,3)\}$. Let us consider the starting configuration $000\;000$
    where the first three components are the original ones and the other three correspond to the
    different elements $E$ in exactly the order presented. If we update the first component, we
    obtain $f_1(\tau^1(000\;000))= f_1(000)=1$ since
    $\tau^1(000\;000)_1=\tau^1(000\;000)_2=\tau^1(000\;000)_3=0$. Then, we reach $100\;000$.
    Updating the $(1,2)$ component, we compute $100\;100$. Updating now the second one, we have
    $f_2(\tau^2(100\;100))= f_1(100)=1$ since $\tau^2(100\;100)_1=1$ and
    $\tau^2(100\;100)_2=\tau^2(100\;100)_3=0$. Then, we reach $110\;100$. Updating the $(2,3)$
    component, we compute $110\;110$. To conclude, we now update the third component, then,
    $f_3(\tau^3(110\;110))= f_1(010)=1$ since $\tau^3(110\;110)_2=1$ and $\tau^3(110\;110)_1=\tau^3(110\;110)_3=0$. 
    Indeed, we are able to reach $111\;110$ given the fact that the third component is able to see the first component at value 0 while the second could see it active.
\end{example}


With our framework, we generalize the vector $V^a$ of the interval update mode to a matrix
$C^a$ in which each element $C^a_{ i,j }$ is equal to the last moment $k$ (with
$k\in\{0,\ldots,a-1\}$) in which the value of the component $j$ has been transmitted to the component $i$.
At first, $C^0$ is a matrix consisting of only zeros.
Remark that several values of a $C^a$ matrix may change from the previous $C^{a-1}$ matrix as this corresponds to propagating changes of different components affecting the component $i$. Moreover, it is possible to have a delay for a node to transmit its next value to itself, i.e. having $C^a_{ i, i } < a-1$ for any component $i$.

\paragraph{Cuttable updates $\Cuttable$:} All trajectories $(x^0, \dots, x^l)$ obtained by $x^0 \in \B^n$ and $w^a = ( i^a, s^a, t^a ) \in [n] \times \B^n \times \B^n$ for $1 \le a \le l$, where
\begin{itemize}
    \item matrix $C^a \in \N^{n \times n}$ such that $C^0 = 0^{n\times n}$, $C^a \geq C^{a-1}$ and, for all $i,
        j \in [n]$, $C^a \leq (a-1)^{n\times n}$;

    \item source $s^a$ with $s^a_j = x^{ C^a_{ i^a, j } }_j$ for all $j \in [n]$;

    \item target $t^a = x^{a-1}$.

\end{itemize}


The correspondence with the initial definition is given in Sect.~\ref{sec:eq-cuttable},
and generalizes the one for the interval updates: the memory-based trajectories correspond to the
projection over the write nodes.

\begin{example} \label{example:cuttable_trajectory}
   
Let us consider the network of the previous example.
The asynchronous graph of $\CuttableNetwork$ is given here below with an example of a $\Cuttable$-trajectory. The same trajectory is also given in Figure \ref{figure:trajectories1} in a different format.

\begin{minipage}[]{.5\linewidth}
    \centering
    \begin{tikzpicture}
        \node (000) at (0,0) {$000$};
        \node (001) at (1,1) {$001$};
        \node (010) at (0,2) {$010$};
        \node (011) at (1,3) {$011$};
        \node (100) at (2,0) {$100$};
        \node (101) at (3,1) {$101$};
        \node (110) at (2,2) {$110$};
        \node (111) at (3,3) {$111$};
    
        \path[draw] (000) -- (001) -- (011) -- (111)
        (000) -- (010) -- (110) -- (111)
        (000) -- (100) -- (101) -- (111)
        (001) -- (101)
        (010) -- (011)
        (100) -- (110);
        
        \draw[thick,-latex, blue] (000) -- (100);
        \draw[thick,-latex, blue] (001) -- (101);
        \draw[thick,-latex, blue] (010) -- (110);
        \draw[thick,-latex, blue] (011) -- (111);
        \draw[thick,-latex, blue] (010) -- (000);
        \draw[thick,-latex, blue] (011) -- (001);
        \draw[thick,-latex, blue] (100) -- (110);
        \draw[thick,-latex, blue] (101) -- (111);
        \draw[thick,-latex, blue] (001) -- (000);
        \draw[thick,-latex, blue] (101) -- (100);
        \draw[thick,-latex, blue] (010) -- (011);
        \draw[thick,-latex, blue] (111) -- (110);
    \end{tikzpicture}
\end{minipage}%
\begin{minipage}[]{.3\linewidth}
    \centering  
    \begin{tabular}{c|c|c|c|c|c}
         $a$    & $i^a$ & $C^a$ & $s^a$ & $t^a$ & $x^a$ \\
         \hline
         $0$ & & $\begin{pmatrix} 0 & 0 & 0 \\ 0 & 0 & 0 \\ 0 & 0 & 0 \end{pmatrix}$ & & & $000$ \\
         $1$    & $1$ & $\begin{pmatrix} 0 & 0 & 0 \\ 0 & 0 & 0 \\ 0 & 0 & 0 \end{pmatrix}$ & $000$ & $000$ & $100$ \\
         $2$    & $2$ & $\begin{pmatrix} 0 & 0 & 0 \\ 1 & 0 & 0 \\ 0 & 0 & 0 \end{pmatrix}$ & $100$ & $100$ & $110$ \\
         $3$    & $3$ & $\begin{pmatrix} 0 & 0 & 0 \\ 1 & 0 & 0 \\ 0 & 2 & 0 \end{pmatrix}$ & $010$ & $110$ & $111$ 
    \end{tabular}
\end{minipage}

\medskip

We can emphasise how this trajectory corresponds to that of the previous example. Indeed the first component to be activated is the first one that, given the local function, cannot do anything other than activate itself. When updating the second component, the information that the first one has been activated can be transmitted. In the previous example, this is done by updating component $(1,2)$, whereas here it is done by setting the corresponding value of the matrix to $1$ to take into account the first component as an active one. 
In both cases, this means that the value of the first component causes the second component to activate.  
When updating the third component, one can decide not to take into account the change of state made by the first component but only consider the second active component. In the previous example, this is done by updating component $(2,3)$ and not updating component $(1,3)$. Here, this corresponds to updating the value of the matrix for the interaction between the second and third components, but not between the first and third. Here too, the result is that the third component changes to $1$.
\end{example}

\begin{figure}
  \centering
  \includegraphics{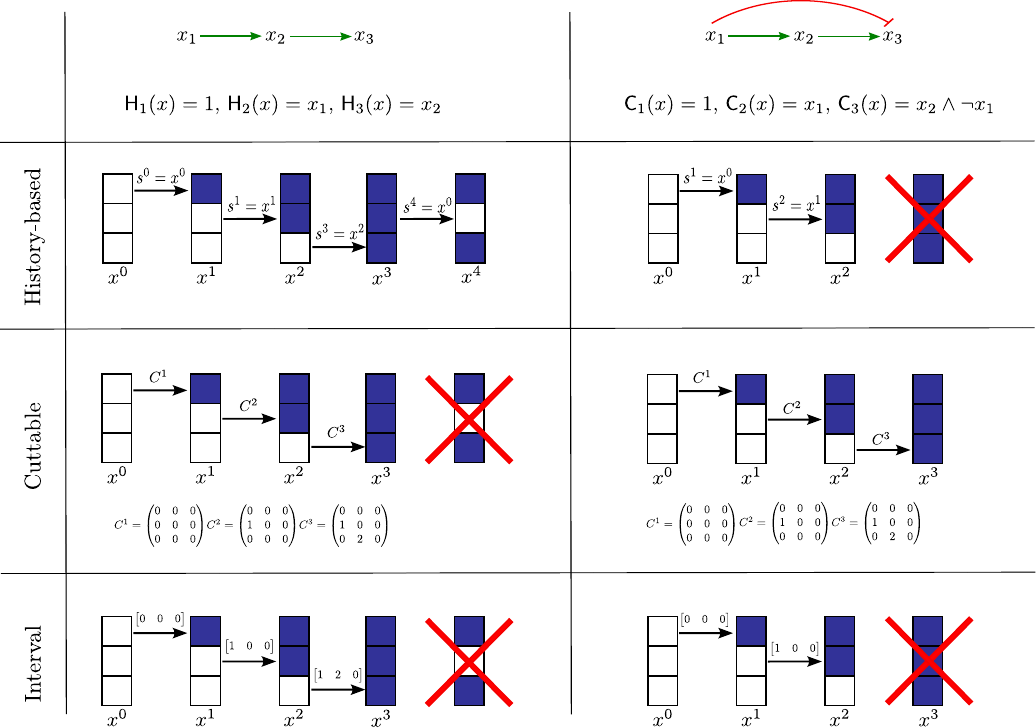}
  \caption{Possible trajectories of certain BNs according to different update modes (see Examples \ref{example:history_trajectory} and \ref{example:cuttable_trajectory}).}
  \label{figure:trajectories1}
\end{figure}

\section{Comparing the different update modes} \label{section:comparing_update_modes}

\subsection{Hierarchy by reachability and trajectory for the different update modes} \label{subsection:hierarchy}

We are interested in comparing the collections of reachability pairs amongst the different update modes listed above. For any update mode $\mu \in \{ \Asynchronous, \History, \TrappingGraph, \MostPermissive, \Subcube, \Interval, \Cuttable \}$, we denote its collection of reachability pairs by $\mu^* = \{ ( x, y ) : x \to_\mu^* y \}$. Then $\nu$ \textbf{weakly simulates} $\mu$ if $\mu^* \subseteq \nu^*$ for all BNs. In Figure \ref{figure:hierarchy}, a strict containment $\mu^* \subsetneq \nu^*$ means that the containment $\mu^* \subseteq \nu^*$ holds for all BNs and that there exists a BN where this containment is strict.

Our hierarchy by reachability is based on a hierarchy by trajectory. As such, we also give the full hierarchy of trajectory in Figure \ref{figure:hierarchy}. Again, a strict containment $\mu \subsetneq \nu$ means that the containment $\mu \subseteq \nu$ holds for all BNs and that there exists a BN where this containment is strict. Clearly, if $\mu \subseteq \nu$, then $\mu^* \subseteq \nu^*$.

\begin{theorem} \label{theorem:hierarchy}
The hierarchy by reachability and trajectory for the update modes listed in Section \ref{section:update_modes} is given in Figure \ref{figure:hierarchy}.
\end{theorem}

\begin{figure}
\resizebox{\textwidth}{!}{
\subfloat[Hierarchy by reachability]{
\begin{tikzpicture}[xscale=3, yscale=2, every node/.style=draw,thick,minimum height=1cm, text width = 3.5cm, align=center]
    \node (A)           at (1,0)  {\small Asynchronous $\Asynchronous^*$};
    \node (Interval)    at (1,1)  {\small Interval $\Interval^*$};
    \node[fill={rgb:black,1;white,4}] (First)       at (0,2)  {\small History-based $\History^*$};
    \node (Cuttable)    at (2,2)  {\small Cuttable $\Cuttable^*$};
    \node (MP)          at (1,3)  {\small Most Permissive $\MostPermissive^*$};
    \node[fill={rgb:black,1;white,4}] (Trapping)    at (1,4)  {\small Subcube-based $\Subcube^*$ \\ Trapping $\TrappingGraph^*$};

    \path[draw] (A) edge (Interval)
    (Interval) edge (First)
    (Interval) edge (Cuttable)
    (First) edge (MP)
    (Cuttable) edge (MP)
    (MP) edge (Trapping);
    
\end{tikzpicture}
} \hspace{1cm}
\subfloat[Hierarchy by trajectory]{
\begin{tikzpicture}[xscale=3, yscale=2, every node/.style=draw,thick,minimum height=1cm, text width = 3.5cm, align=center]
    \node (Asynchronous)    at (1,0) {\small Asynchronous $\Asynchronous$};
    \node[fill={rgb:black,1;white,4}] (First)       at (0,1) {\small History-based $\History$};
    \node (Interval)    at (2,1) {\small Interval $\Interval$};
    \node (Cuttable)    at (2,2) {\small Cuttable $\Cuttable$};
    \node (MP)          at (2,3) {\small Most Permissive $\MostPermissive$};
    \node[fill={rgb:black,1;white,4}] (Trapping)    at (0,3) {\small Trapping $\TrappingGraph$};
    \node[fill={rgb:black,1;white,4}] (TrappingH)    at (1,4) {\small Subcube-based $\Subcube$};

    \path[draw]
    (Asynchronous) edge (Interval)
    (Asynchronous) edge (First)
    (First) edge (MP)
    (Cuttable) edge (MP)
    (Interval) edge (Cuttable)
    (First) edge (Trapping)
    (MP) edge (TrappingH)
    (Trapping) edge (TrappingH);
    
\end{tikzpicture}
}
}
    \caption{Hierarchy of different update modes. A darker background is used to highlight the update modes introduced for the first time in this work.}
    \label{figure:hierarchy}
\end{figure}


The rest of this subsection is devoted to the proof of Theorem \ref{theorem:hierarchy}, broken down into items \ref{item:a} to \ref{item:j}. The structure of the proof is displayed in Figure \ref{figure:proof_hierarchy}.

\begin{figure}
\subfloat[Hierarchy by reachability]{
\begin{tikzpicture}[xscale=3, yscale=2]
    \node (A)           at (1,0)  {\small $\Asynchronous^*$};
    \node (Interval)    at (1,1)  {\small $\Interval^*$};
    \node (History)       at (0,2)  {\small $\History^*$};
    \node (Cuttable)    at (2,2)  {\small  $\Cuttable^*$};
    \node (MP)          at (1,3)  {\small $\MostPermissive^*$};
    \node (Trapping)    at (1,4)  {\small $\TrappingGraph^*$};
    \node (Subcube)     at (1,5)  {\small $\Subcube^*$};

    \path[draw, blue, -latex] 
    (A) edge [bend right] node[right] {\ref{item:c}} (Interval)
    (Interval) edge node[left] {\ref{item:e}} (History)
    (Interval) edge node[right] {\ref{item:c}} (Cuttable)
    (History) edge node[left] {\ref{item:b}} (MP)
    (Cuttable) edge node[right] {\ref{item:c}} (MP)
    (MP) edge[bend right] node[right] {\ref{item:b}} (Trapping)
    (Trapping) edge[bend right] node[right] {\ref{item:a}} (Subcube)    
    (Subcube) edge[bend right] node[left] {\ref{item:d}} (Trapping);

    \path[draw, red, dotted, bend right, -latex] 
    (Trapping) edge node[left] {\ref{item:g}} (MP)
    (Cuttable) edge node[above] {\ref{item:i}} (History)
    (History) edge node[below] {\ref{item:h}} (Cuttable)
    (Interval) edge node[left] {\ref{item:j}} (A);

\end{tikzpicture}
} \hspace{1cm}
\subfloat[Hierarchy by trajectory]{
\begin{tikzpicture}[xscale=3, yscale=2]
    \node (Asynchronous)    at (1,0) {\small  $\Asynchronous$};
    \node (History)       at (0,1) {\small $\History$};
    \node (Interval)    at (2,1) {\small $\Interval$};
    \node (Cuttable)    at (2,2) {\small $\Cuttable$};
    \node (MP)          at (2,3) {\small $\MostPermissive$};
    \node (Trapping)    at (0,3) {\small $\TrappingGraph$};
    \node (Subcube)    at (1,4) {\small $\Subcube$};

    \path[draw, blue, -latex]
    (Asynchronous) edge node[right] {\ref{item:c}} (Interval)
    (Asynchronous) edge node[left] {\ref{item:a}} (History)
    (History) edge [bend left] node[left] {\ref{item:b}} (MP)
    (Cuttable) edge node[right] {\ref{item:c}} (MP)
    (Interval) edge node[right] {\ref{item:c}} (Cuttable)
    (History) edge node[left] {\ref{item:a}} (Trapping)
    (MP) edge node[right] {\ref{item:b}} (Subcube)
    (Trapping) edge node[left] {\ref{item:a}} (Subcube);

    \path[draw, dotted, red, -latex]
    (Interval) edge [bend right] node[right] {\ref{item:f}} (Trapping)
    (Trapping) edge node[above] {\ref{item:g}} (MP)
    (History) edge [bend right] node[below] {\ref{item:h}} (Cuttable)
    (Cuttable) edge [bend right] node[above] {\ref{item:i}} (History);
    
\end{tikzpicture}
}
    \caption{Structure of the proof of Theorem \ref{theorem:hierarchy}. On the edges, the letters identifying the corresponding part of the paper are indicated. Solid blue edges correspond to inclusions, while dotted red edges correspond to non-inclusions.}
    \label{figure:proof_hierarchy}
\end{figure}

\begin{enumerate}[(a)]
    \item \label{item:a} $\Asynchronous \subseteq \History \subseteq \TrappingGraph \subseteq \Subcube$ and hence $\Asynchronous^* \subseteq \History^* \subseteq \TrappingGraph^* \subseteq \Subcube^*$.
\end{enumerate}

This follows from the definitions of the different update modes.

\begin{enumerate}[(a), resume]
    \item \label{item:b} $\History \subseteq \MostPermissive \subseteq \Subcube$ and hence $\History^* \subseteq \MostPermissive^* \subseteq \Subcube^*$.
\end{enumerate}

Again, this follows from the definitions of the different update modes.

\begin{enumerate}[(a), resume]
    \item \label{item:c} $\Asynchronous \subseteq \Interval \subseteq \Cuttable \subseteq \MostPermissive$ and hence $\Asynchronous^* \subseteq \Interval^* \subseteq \Cuttable^* \subseteq \MostPermissive^*$. 
\end{enumerate}

Indeed, an $\Asynchronous$-trajectory is an $\Interval$-trajectory with $V^a_j = a-1$ for all $a \ge 1$ and $j \in [n]$; an $\Interval$-trajectory is a $\Cuttable$-trajectory with $C^a_{ j, i } = V^a_j$ for all $0 \le a \le l$ and $j,i \in [n]$; and a $\Cuttable$-trajectory is an $\MostPermissive$-trajectory since $s^a \in [ x^0, \dots, x^{a-1} ]$ for all $1 \le a \le l$.


\begin{enumerate}[(a), resume]
    \item \label{item:d} $\TrappingGraph^* = \Subcube^*$.
\end{enumerate}

We now prove that the trapping update mode and the subcube-based update mode allow exactly to reach any configuration in the whole principal trapspace.

\begin{proposition} \label{proposition:trapping_computing_with_memory}
For all $f \in \Functions(n)$, and all $x,y \in \B^n$, the following are equivalent:
\begin{enumerate}
    \item \label{item:reachable_trapping}
    $y$ is reachable from $x$ by trapping updates, i.e. $x \to^*_\TrappingGraph y$; 

    \item \label{item:reachable_subcube}
    $y$ is reachable from $x$ by subcube-based updates, i.e. $x \to^*_{\Subcube} y$;

    
    \item \label{item:in_trapspace}
    $y \in T_f(x)$.
\end{enumerate}
\end{proposition}

\begin{proof}
$\ref{item:reachable_trapping} \implies  \ref{item:reachable_subcube}$. Trivial.

$\ref{item:reachable_subcube} \implies \ref{item:in_trapspace}$. Let $x = x^0 \to_\Subcube \dots \to_\Subcube x^l = y$. We prove that $x^a \in T_f(x)$ by induction on $0 \le a \le l$. The case $a = 0$ is trivial, therefore suppose it holds for up to $a-1$. We have $s^a, t^a \in [ x^0, \dots, x^{a-1} ] \subseteq T_f(x)$. Let $u^a = f( s^a )$, then again $u^a \in T_f( x )$ since $T_f( x )$ is a trapspace. We obtain
\[
    x^a = ( f_{i^a}( s^a ), t^a_{ -i^a } ) = ( u^a_{i^a}, t^a_{ -i^a } ) \in [ u^a, t^a ] \subseteq T_f(x).
\]


$\ref{item:in_trapspace} \implies \ref{item:reachable_trapping}$.
We prove that if a configuration $y$ (hence $x$) is reachable from $x$ by trapping updates, then
so are all the configurations in $T_f(y)$ (hence $T_f(x)$).
We first prove that if $y$ is reachable from $x$ by trapping updates, then any $z \in [x,y]$ is reachable from $x$. Let $x = x^0 \to_\TrappingGraph \dots \to_\TrappingGraph x^l = y$, then for all $j \in \Delta( x, y )$, there exists $0 \le b(j) \le l-1$ such that $y_j = f_j( x^{b(j)} )$ (since the source $s^a$ always belongs to $\{ x^0, \dots, x^{l-1} \}$). Therefore, one can reach $z$ by extending the trajectory as follows: if $\Delta( x, z ) = \{ j_1, \dots, j_k \}$, let $w^{l+a} = ( j_a, b(j_a), x )$ for $1 \le a \le k$, then $x^{l+k} = z$. Therefore, the set $\{ y : x \to_{\TrappingGraph}^* y \}$ of configurations reachable from $x$ is a subcube. We now prove that if $y$ is reachable from $x$, then $f(y)$ is also reachable from $x$. Let $x = x^0 \to_\TrappingGraph \dots \to_\TrappingGraph x^l = y$ and define $w^{l+j} = ( i^{l+j} = j, s^{l+j} = y, t^{l+j} = x^{l+j-1} )$ for all $j \in [n]$; it is then easy to verify that $x^{l+n} = f(y)$. Thus $\{ y : x \to_{\TrappingGraph}^* y \}$ is a trapspace containing $x$, and hence it contains $T_f( x )$.
\end{proof}

\begin{enumerate}[(a), resume]
    \item \label{item:e} $\Interval^* \subseteq \History^*$.
\end{enumerate}

We first prove two lemmas about history-based and interval trajectories, respectively.

\begin{lemma} \label{lemma:repetition}
If $x^0 \to_\History \dots \to_\History x^l$ is a history-based trajectory, then so is $x^0 \to_\History \dots \to_\History x^l \to_\History x^{l+1} = x^l$.
\end{lemma}

\begin{proof}
Let $w^l = (i^l, s^l = x^{l'}, t^l = x^{l-1})$ where $l' < l$. Define $w^{l+1} = (i^{l+1} = i^l, s^{l+1} = s^l, t^{l+1} = x^l)$; we then have $x^{l+1} = ( f_{ i^{l+1} }( s^{l+1} ) = x^l_{ i^l }, t^{ l+1 }_{ -i^{l+1} } = x^l_{ -i^l } ) = x^l$.
\end{proof}

\begin{lemma}\label{lemma:distance}
    If $x \to_\Interval^* y$, there is always an interval trajectory $x = x^0 \to_\Interval x^1 \to_\Interval \cdots \to_\Interval x^p = y$ with $\HammingDistance(s^{a-1},s^a) \le 1$ for all $1\leq a \leq p$.
\end{lemma}

\begin{proof}
Let $x = x^0 \to_\Interval \dots \to_\Interval x^l = y$ be an interval trajectory where $\HammingDistance( s^{a-1}, s^a ) \le 1$ for all $1 \le a \le p-1$ and $\HammingDistance( s^{p-1}, s^p ) > 1$. Let us prove how we can replace the transition $x^p \to_\Interval x^{p+1}$ with a sequence of transitions $u^0 = x^p \to_\Interval \dots \to_\Interval u^n = x^{p+1}$ and still have an interval trajectory. 

For any $0 \le i \le n$, let $v^i$ be the configuration such that $v^i_j = s^p_j$ for all $1 \le j \le i$ and $v^i_j = s^{p-1}_j$ for all $i+1 \le j \le n$. We then have $\HammingDistance( v^i, v^{i-1} ) \le 1$ for all $1 \le i \le n$. Let $u^i$ be recursively defined as $u^0 = x^p$ and $u^i = ( f_i( v^{i-1} ), u^{i-1}_{ -i } )$. 

Denote $\dot{x}^a = x^a$ for all $0 \le a \le p$, $\dot{x}^{p+i} = u^i$ for all $1 \le i \le n$, $\dot{x}^{a+n} = x^{a+1}$ for all $a \ge p+1$. We prove that
\[
    \dot{x}^0 = x^0 \to \dots \to \dot{x}^p = x^p = u^0 \to \dots \to \dot{x}^{p+n} = u^n = x^{p+1} \to \dots \to \dot{x}^{l+n-1} = x^l
\]
is an admissible interval trajectory. The first part, $\dot{x}^0 \to_\Interval \dots \to_\Interval \dot{x}^p$, follows from our hypothesis. For the second part, $\dot{x}^p \to_\Interval \dots \to_\Interval \dot{x}^{p+n}$, use $w^{p+i} = (i^p, v^{i-1}, u^{i-1})$. For the third part, $\dot{x}^{p+n} \to_\Interval \dots \to_\Interval \dot{x}^{n+1-l}$, then $\dot{w}^{n-1+a} = w^a$ is a valid triple for interval updates.
\end{proof}

We can now prove the main result.

\begin{proposition}\label{pr:IntVSHist}
    Given two configurations $x$ and $y$, if $x \to_\Interval^* y$, then $x \to_\History^* y$.
\end{proposition}

\begin{proof}
    Let $x=x^0 \to_\Interval x^1 \to_\Interval \cdots \to_\Interval x^l=y$ be an interval trajectory with corresponding triples $w^a = ( i^a, s^a, t^a)$ for all $1 \leq a\leq l$.
    According to Lemma \ref{lemma:distance}, we assume $\HammingDistance(s^{a-1},s^a) \le 1$ for all $1\leq a \leq l$.

    We first prove, by induction over $1 \le a \le l$, that there always exists a history-based trajectory $x = x^0 \to_\History s^1 \to_\History \cdots \to_\History s^a$. 
   The basic case $a = 1$ is trivial because $s^1 = x^0$.
    Assume the statement holds for $a-1$, i.e. $x = x^0 \to_\History s^1 \to_\History \cdots \to_\History s^{a-1}$ is a history-based trajectory. Lemma \ref{lemma:repetition} settles the case where $s^a = s^{a-1}$, hence we assume $s^{a-1} \neq s^a$ and, according to Lemma \ref{lemma:distance}, $\Delta( s^{a-1}, s^a) = \{ j \}$. 
    Then, for $q = V^a_j$, we have $q < a$ and $x^q_j = s^a_j$.
    Suppose, for the sake of contradiction, that $x^m_j=x^q_j$ for all $0\leq m\leq q $.
    Then, $s^m_j=x^q_j$ for all $0\leq m \leq q$, and in particular $s^{a-1}_j = s^a_j$, which is the desired contradiction.
    Therefore, there exists $q'<q$ such that $f_j(s^{q'})=x^q_j$.
    Thus, let $s^a=s^{q'}$ and $i^a=j$. We obtain $(f_j(s^{q'})=s^a_j, s^{a-1}_{-j})=s^a$.

    We now prove that we can extend the history-based trajectory all the way to $y$. 
    More formally, we prove that  
    $x=x^0 \rightarrow s^1 \rightarrow \cdots \rightarrow s^l \rightarrow \cdots \rightarrow y$ is a history-based trajectory, where $s^l = g^l \rightarrow \cdots \to g^m \to \cdots \to  g^{l + \mid\Delta\mid} = y$ is any geodesic.
    Let $\Delta = \Delta( s^l, y)$ and let $i^m$ the coordinate where $g^m$ and $g^{m-1}$ differ for all $l+1\leq m\leq l+ \mid\Delta\mid$, so that $\{i^{l+1},\ldots,i^{l+\mid\Delta\mid}\}=\Delta$. For all  $j \in \Delta$, 
    since $s_j^l = x^b_j \neq y_j$ (for some $b < l$), there exists $q_j$ such that $j=i^{q_j}$ and $f_j(s^{q_j})= y_j$. 
    Then $x=x^0 \rightarrow s^1\rightarrow \cdots \rightarrow s^l \rightarrow g^{l+1} \to \cdots \rightarrow g^{l+\mid\Delta\mid}=y$ is a history-based trajectory with $(i^m, s^m, t^m)$ for all $l+1\leq m\leq l+ \mid\Delta\mid$, where $s^m=s^{q_{i^m}}$ and $t^m = g^{m-1}$. 
\end{proof}

\begin{enumerate}[(a), resume]
    \item \label{item:f} $\Interval \not\subseteq \TrappingGraph$.
\end{enumerate}

\begin{proof}[\proofname\ (An interval trajectory that is not trapping)]\label{ex:interval-trapping}
Let $\IntervalNetwork$ be the network from Example \ref{example:interval_trajectory}. Then we claim that $(000, 100, 101, 111, 011)$ is an $\Interval$-trajectory but not a $\TrappingGraph$-trajectory.

We have already shown that it is an $\Interval$-trajectory. It is not a $\TrappingGraph$-trajectory, as $x_1^4 \ne x^3_1$ implies that $i^4 = 1$, while $\IntervalNetwork_1( x^0 ) = \IntervalNetwork_1( x^1 ) = \IntervalNetwork_1( x^2 ) = \IntervalNetwork_1( x^3 ) = 1$ implies that $s^4 \notin \{ x^0, \dots, x^3 \}$.
\end{proof}

\begin{enumerate}[(a), resume]
    \item \label{item:g} $\TrappingGraph^* \not\subseteq \MostPermissive^*$ and hence $\TrappingGraph \not\subseteq \MostPermissive$.
\end{enumerate}

All the update modes considered in this paper have the property that trajectories can be
\emph{compressed}.
Our definition of trajectory allows for duplicates, i.e. transitions $x^{a-1} \to_\mu x^a$ where $x^{a-1} = x^a$; a simple example is when $x^{a-1}$ is a fixed point of $f$ and $s^a = t^a = x^{a-1}$. The update modes in this paper are then closed by removing duplicates.

Formally, let us call a trajectory $y^0 \to_\mu \dots \to_\mu y^l$ \Define{compressed} if $y^a \ne y^{a-1}$ for all $1 \le a \le l$. For any trajectory, its \Define{compressed version} is obtained by removing all duplicates; clearly the latter is indeed compressed.

\begin{lemma} \label{lemma:compression}
For all $\mu \in \{ \Asynchronous, \History, \TrappingGraph, \MostPermissive, \Subcube, \Interval, \Cuttable \}$, if $x^0 \to_\mu \dots \to_\mu x^l$ is a $\mu$-trajectory, then so is its compressed version.
\end{lemma}

\begin{proof}
We only need to prove that one can remove one duplicate of a $\mu$-trajectory and still obtain a $\mu$-trajectory. Let the $\mu$-trajectory $x^0 \to_\mu \dots x^{b-1} \to_\mu x^b \to_\mu x^l$ be described by the triples $w^a = ( i^a, s^a, t^a )$ for all $1 \le a \le l$, and let the duplicate be $x^b = x^{b-1}$. Consider the $\mu$-trajectory $\hat{x}^0 \to_\mu \dots \to_\mu \hat{x}^{l-1}$ described by $\hat{w}^a = ( \hat{i}^a, \hat{s}^a, \hat{t}^a )$, where
\[
    ( \hat{i}^a, \hat{s}^a, \hat{t}^a ) = \begin{cases}
    ( i^a, s^a, t^a ) &\text{if } a \le b-1 \\
    ( i^{a+1}, s^{a+1}, t^{a+1} ) &\text{if } a \ge b.
    \end{cases}
\]
The beginning of the trajectory satisfies $\hat{x}^0 \to_\mu \dots \to_\mu \hat{x}^{b-1} = x^0 \to_\mu \dots x^{b-1}$. Then, by induction on $b-1 \le a \le l-1$, we can easily verify that $\hat{x}^a = x^{a-1}$ and that $\hat{x}^0 \to_\mu \dots \to_\mu \hat{x}^a$ is indeed a $\mu$-trajectory. We conclude that removing the duplicate $x^b = x^{b-1}$ still yields a $\mu$-trajectory.
\end{proof}

\begin{corollary} \label{corollary:compression}
If there is a $\mu$-trajectory from $x$ to $y$, then there is a compressed $\mu$-trajectory from $x$ to $y$.
\end{corollary}

For this and the following items, we aim to prove that for a given pair $x, y$ of configurations, there is no $\mu$-trajectory from $x$ to $y$ for a particular update mode $\mu$ (for this item, the most permissive update mode). Based on Corollary \ref{corollary:compression}, we always prove that there is no compressed $\mu$-trajectory from $x$ to $y$.
 
\begin{proof}[\proofname\ (A trapping reachability pair that is not most permissive)] \label{ex:trapping-MP}
Let $\TrappingNetwork$ be the network from Example \ref{example:trapping_trajectory}. We claim that $(00, 01)$ is a $\TrappingGraph$-reachability pair but not a $\MostPermissive$-reachability pair. First, we have already given a $\TrappingGraph$-trajectory from $00$ to $01$. Second, any $\MostPermissive$-trajectory starting at $x^0 = 00$ must have $x^1 = 10$. Therefore, any trajectory reaching $y = 01$ must have some $a \ge 2$ with $i^a = 1$ and $\TrappingNetwork_1( s^a ) = 0$, which is the desired contradiction.
\end{proof}


\begin{enumerate}[(a), resume]
    \item \label{item:h} $\History^* \not\subseteq \Cuttable^*$ and hence $\History \not\subseteq \Cuttable$.
\end{enumerate}

\begin{proof}[\proofname\ (A history-based reachability pair that is not cuttable)]\label{ex:first-cuttable}
    Let us consider the network $\HistoryNetwork$ from Example \ref{example:history_trajectory}, with $\HistoryNetwork_1(x) = 1$, $\HistoryNetwork_2(x) = x_1$ and $\HistoryNetwork_3(x) = x_2$. We have shown that $y = 101$ is $\History$-reachable from $x = 000$; let us now prove that $y$ is not $\Cuttable$-reachable from $x$.

    Suppose that there is a cuttable trajectory $x = x^0 \to_\Cuttable \dots \to_\Cuttable x^l = y$. We start with $x^0 = 000$ with $C^0_{ i, j }=0$ for all $i, j$.
    In this network, we must activate the second to activate the third component.
    For this reason, the first component must be updated.
    Then, we consider $i^1 = 1$ and $C^1 = C^0$.
    Accordingly, $s^1 = x^0$ and $x^1 = 100$.
    Afterwards, to update the second component (i.e., $i^2 = 2$), we must consider $C^2_{2, 1} = 1$ and $C^2_{i, j} = C^1_{i, j}$ otherwise.
    As a result, $s^2 = x^1_1 x^0_2 x^0_3 = 100$ and $x^2 = 110$.
    Later, $x^3$ can be calculated with $i^3 = 3$, $C^3_{3, 2}=2$ and $C^3_{i, j} = C^2_{i, j}$ otherwise.
    Accordingly, $s^3 = x^1_1x^2_2x^0_3 = 110$ and $x^3=111$.
    At this point, to reach the configuration $101$, we should be able to update the second component by considering a matrix $C^4$ where $C^4_{2, 1} = 0$ (to deactivate the second component). 
    However, this is not possible as $C^4 \geq C^3$.
    In conclusion, $101$ cannot be reached by $000$.

    This argument is displayed in the left column of Figure \ref{figure:trajectories1}.
\end{proof}

\begin{enumerate}[(a), resume]
    \item \label{item:i} $\Cuttable^* \not\subseteq \History^*$.
\end{enumerate}

\begin{proof}[\proofname\ (A cuttable reachability pair that is not history-based)]\label{ex:cuttable-first}
    Let us consider the network $\CuttableNetwork$ from Example \ref{example:cuttable_trajectory} with $\CuttableNetwork_1(x)=1$, $\CuttableNetwork_2(x)= x_1$ and $\CuttableNetwork_3(x)=x_2 \land \neg x_1$.

    We have shown that $y = 111$ is $\Cuttable$-reachable from $x = 000$; let us now prove that $y$ is not $\History$-reachable from $x$.
    According to the history-based updates, we can reach $x^1=100$ (with $s^0 = x^0$ and $i^1 = 1$) and $x^2 = 110$ (with $s^2 = x^1$ and $i^2 = 2$).
    However, it is impossible to reach the configuration $010$ (the only configuration that allows the third component to be activated).
    We can therefore conclude that $111$ is unreachable.

    This argument is displayed in the right column of Figure \ref{figure:trajectories1}.
\end{proof}

\begin{enumerate}[(a), resume]
    \item \label{item:j} $\Interval^* \not\subseteq \Asynchronous^*$. 
\end{enumerate}

\begin{proof}[\proofname\ (An interval reachability pair that is not asynchronous)] \label{ex:interval-asynchronous}
Let us consider the network $\IntervalNetwork$ from Example \ref{example:interval_trajectory}. We have shown that $y = 011$ is reachable from $x = 000$ by interval updates. However, it is easily seen that $011$ is not reachable from $000$ by asynchronous updates.
\end{proof}

\subsection{Refinement of the hierarchy by reachability} \label{subsection:refinement}

We make two notes about the hierarchy by reachability in Figure \ref{figure:hierarchy}.

The hierarchy by reachability gives $\Cuttable^* \subseteq \MostPermissive^*$ and $\History^* \subseteq \MostPermissive^*$, or in other words a reachability pair that is either history-based or cuttable is most permissive. We now prove that there are most permissive reachability pairs that are neither history-based nor cuttable, i.e. $\MostPermissive^* \ne (\History^* \cup \Cuttable^*)$. Let $\CuttableNetwork \in \Functions(3)$ from the proof of (\textit{i}) and $\HistoryNetwork$ from the proof of (\textit{h}), and let $f \in \Functions(6)$ be defined by
\[
    f( x_{123}, x_{456} ) = ( \CuttableNetwork( x_{123} ), \HistoryNetwork( x_{456} ) ).
\]
Then $(000000, 111101)$ is a reachability pair that is most permissive but neither history-based nor cuttable.

Similarly, the hierarchy by reachability gives $\Interval^* \subseteq \History^*$ and $\Interval^* \subseteq \Cuttable^*$, or in other words any interval reachability pair is both history-based and cuttable. We now prove that there are some reachability pairs are both history-based and cuttable but not interval, i.e. $\Interval^* \ne (\History^* \cap \Cuttable^*)$. We use a similar example to that above -- as each trajectory is not an interval trajectory, but this time, we use a middle component that gets activated when either trajectory is used. The details are given in Figure \ref{figure:history_cap_cuttable}.

\begin{figure}
  \centering
  \includegraphics[width=0.9\textwidth]{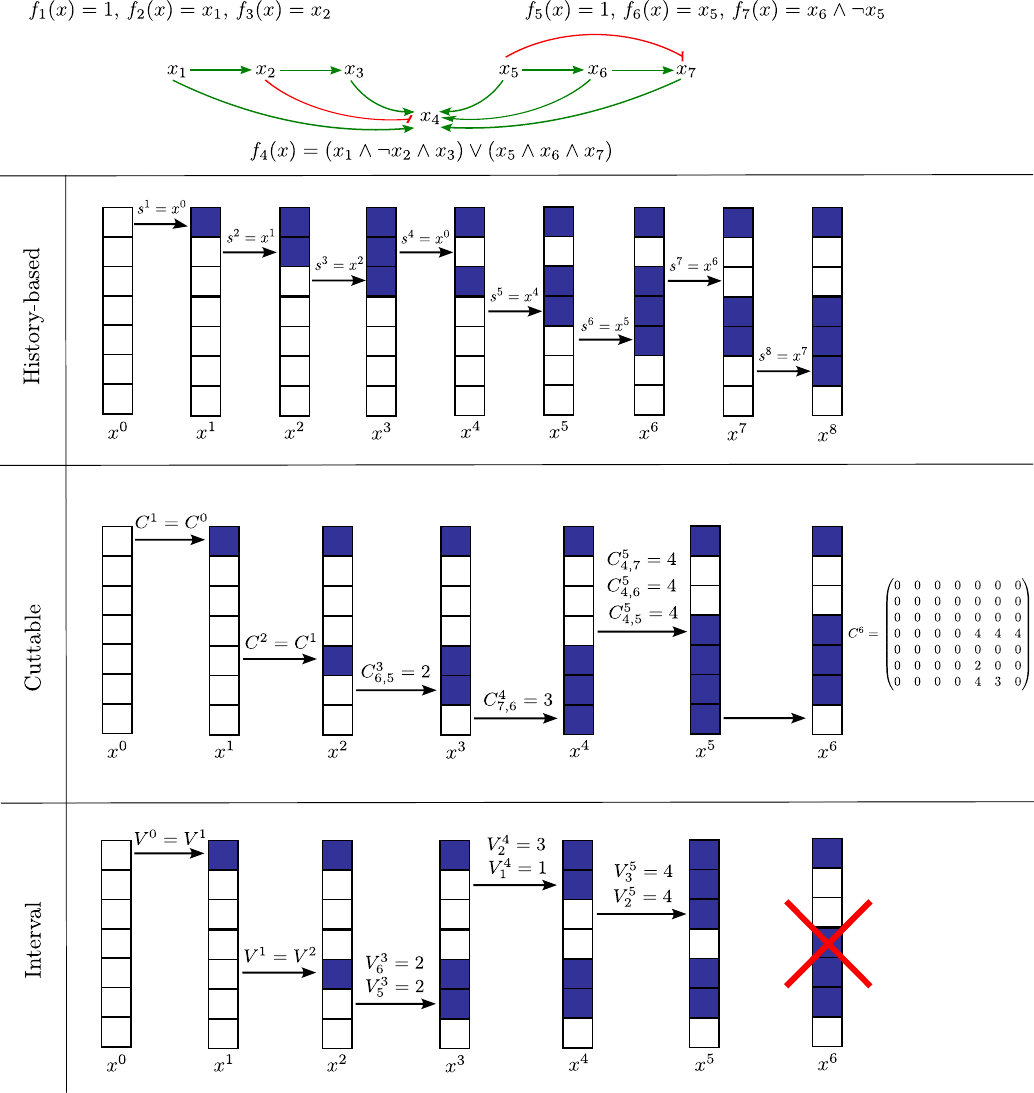}
  \caption{Example of a history-based and cuttable trajectory that is not interval. The example is also valid concerning reachability.}
  \label{figure:history_cap_cuttable}
\end{figure}

\subsection{Equivalent update modes} \label{subsection:commensurate}

In Section \ref{subsection:hierarchy}, when discussing reachability we only considered the inclusion relation $\mu^* \subseteq \nu^*$; here we consider two equivalence relations for update modes based on reachability. First, for any $x$, let $\mu^*(x) = \{ y : x \to_\mu^* y \}$ be the set of configurations reachable from $x$. Then say two update modes are \Define{commensurate} if there exist two functions $\phi, \psi : \N \to \N$ such that for all BNs $f$ and all configurations $x$, 
\[
    | \mu^*(x) | \le \phi( | \nu^*(x) | ), \qquad | \nu^*(x) | \le \psi( | \mu^*(x) | ).
\]
In other words, the number of reachable configurations in one mode gives an estimate of the number of reachable configurations in the other. Second, we say that a configuration $y$ is a \Define{min-trapspace configuration} if its principal trapspace is minimal. We say $\mu$ and $\nu$ are \Define{min-trapspace-equivalent} if for all BNs $f$, all configurations $x$, and all min-trapspace configurations $y$, $x \to_\mu^* y$ if and only if $x \to_\nu^* y$.

We now classify the commensurate and min-trapspace-equivalent update modes from our list, and show
these two notions are equivalent. Since $\TrappingGraph^* = \Subcube^*$, we omit $\Subcube$.

\begin{theorem} \label{theorem:commensurate}
Let $\mu, \nu \in \{ \Asynchronous, \History, \TrappingGraph, \MostPermissive, \Interval, \Cuttable \}$. The following are equivalent:
\begin{enumerate}
    \item $\mu$ and $\nu$ are commensurate;

    \item $\mu$ and $\nu$ are min-trapspace-equivalent;

    \item either $\mu = \nu$ or $\{ \mu, \nu \} = \{ \TrappingGraph, \MostPermissive \}$.
\end{enumerate}
\end{theorem}

We first prove that $\MostPermissive$ and  $\TrappingGraph$ are min-trapspace-equivalent
(Lemma~\ref{lemma:min-trapspace_configurations})
and commensurate
(Proposition~\ref{proposition:trapping_MP_commensurate}).

Then, using examples, we show that:
(1) that the most permissive/history-based update mode is neither min-trap-equivalent nor commensurate to the cuttable/interval update mode. (Example \ref{example:H_C_not_commensurate});
(2) the most permissive/cuttable update mode is neither min-trap-equivalent nor commensurate to the history-based/interval update mode (Example \ref{example:C_H_not_commensurate});
(3) the interval and asynchronous update modes are neither min-trapspace-equivalent nor commensurate (Example \ref{example:I_A_not_commensurate}).

\begin{lemma} \label{lemma:min-trapspace_configurations}
For any configuration $x$ and any min-trapspace configuration $y \in T_f(x)$ (i.e., $x \to^*_\TrappingGraph y$), we have $x \to_\MostPermissive^* y$.
\end{lemma}

\begin{proof}
We first prove that $x \to_\MostPermissive^* u := T_f(x) - x$. Let $z \in T_f( x )$ be the furthest configuration reachable from $x$ in $\MostPermissive$, and let $x = x^0 \to_\MostPermissive \dots \to_\MostPermissive x^{l-1} = z$. For the sake of contradiction, assume $z \ne u$. Since $[x, z]$ is a strict subcube of $T_f(x)$, it is not a trapspace, hence there exists $s^l \in [x, z] \subseteq [x^0, \dots, x^{l-1}]$ such that $f( s^l ) \notin [x, z]$. In particular, there exists a component $i^l \in [n]$ such that $x_{i^l} = z_{i^l} \ne f_{i^l}( s^l )$. Letting $w^l = ( i^l, s^l, t^l = z )$, we obtain $x^l$ that is reachable from $x$ and yet is further away from $x$ than $z$ is, which is the desired contradiction.

Let $x^0 = x \to_\MostPermissive \dots \to_\MostPermissive x^l = u$. Now, for all $j \in [n]$, there exists $s^{ l+j } \in T_f( y )$ with $f_j( z ) = y_j$. Indeed, suppose, for the sake of contradiction, that $f_i(z) = \neg y_i$ for all $z \in T_f(y)$. Then the subcube $T_f( y ) \cap \{ x \in \B^n : x_i = \neg y_i \}$ is a trapspace that is strictly contained in $T_f( y )$, which is the desired contradiction.

Therefore, the trajectory $x^0 = x \to_\MostPermissive \dots \to_\MostPermissive x^l = u = u^0 \to_\MostPermissive u^1 \to_\MostPermissive \dots \to_\MostPermissive u^n = y$ is indeed an $\MostPermissive$-trajectory with $w^{l + j} = ( j, s^{l+j}, u^{l-1} )$ for all $1 \le j \le n$.
\end{proof}

We now prove that trapping and most permissive update modes are commensurate. In fact, we can be a lot more precise.

\begin{proposition} \label{proposition:trapping_MP_commensurate}
For any configuration $x$ with $| \TrappingGraph^*(x) | = 2^d$, we have $| \MostPermissive^*(x) | \ge L(d) = 2^{ \lfloor d/2 \rfloor } + 2^{ \lceil d/2 \rceil } - 1$. Conversely, for any $L(d) \le k \le 2^d$, there exists $f \in \Functions( d )$ and $x \in \B^d$ such that $T_f(x) = \B^d$ while $| \MostPermissive^*(x) | = k$.
\end{proposition}

\begin{proof}
Without loss, let $d = n$ and $x = 0^n$. Let $C$ be the set of local functions $f_i$ that are constant (equal to $1$) in $T_f(x) = \B^n$, and let $c = |C|$ and $D = [n] \setminus C$. Then $x$ can always reach the following two sets of configurations in $\MostPermissive$. 
\begin{enumerate}
    \item $Y = \{ y \in \B^n : y_C = 1 \}$. The proof is similar to that of Lemma \ref{lemma:min-trapspace_configurations}: first reach $u = T_f(x) - x = 1 \dots 1$, then for all $i \ge c+1$, there exists $s^i$ such that $f_i( s^i ) = y_i$, hence $y$ can be reached.

    \item $Z = \{ z \in \B^n : z_D = 0 \}$. Any geodesic from $x$ to $z$ is a most permissive trajectory.
\end{enumerate}
We have $|Y| = 2^{n-c}$, $|Z| = 2^c$, and $| Y \cap Z | = 1$, thus $x$ can reach at least $2^{n-c} + 2^c - 1 \ge L(d)$ configurations.

We prove, by induction on $n$, that for any $L(n) \le k \le 2^n$ there exists $f \in \Functions(n)$ and $x \in \B^n$ such that $T_f(x) = \B^n$ and $| \MostPermissive^*(x) | = k$. We also use its immediate consequence: for any $n \ge 1$ and any $1 \le r \le 2^n$, there exists $f \in \Functions(n)$ and $x \in \B^n$ such that $| \MostPermissive^*(x) | = r$. The claim is trivial for $n=1$ so let us assume it holds up to $n-1$. Let $c = \lfloor n/2 \rfloor$ and define $q$ and $r$ as the quotient and remainder of the following long division:
\[
    k + 2^{n-c} - 2^c = q(2^{n-c} - 1) + r. 
\]
Let $Q$ be an up-set of $\B^c$ (i.e. if $s \in Q$ and $t \ge s$, then $t \in Q$) with $|Q| = q$ and let $q^*$ be a minimal element of $Q$.

By induction hypothesis, let $g \in \Functions(n-c)$ and $z = 0^{n-c} \in \B^{n-c}$ such that $| \MostPermissive^*( z ) | = r$.
Let $C = \{1, \dots, c \}$ and $D = \{c+1, \dots, n\}$ and let $f \in \Functions(n)$ be defined as
\begin{align*}
    f_C( x ) &= 1_C, \\
    f_D(x) &= \begin{cases}
    \neg x_D &\text{if } x_C \in Q \setminus \{ q^* \} \\
    g( x_D ) &\text{if } x_C = q^* \\
    x_D      &\text{if } x_C \notin Q.
    \end{cases}
\end{align*}
Since the local functions in $C$ are constant and equal to $1$, any $\MostPermissive$-trajectory is monotone with respect to $x_C$, i.e. if $x^0 \to_\MostPermissive \dots \to_\MostPermissive x^l$, then $x^0_C \le \dots \le x^l_C$. For any $\alpha \in \B^c$, let $R_\alpha = \{ y \in \MostPermissive^*(x) : y_C = \alpha \}$; then $R_\alpha = \bigcup_{\beta \le \alpha} R_\beta$. 

We can now determine the configurations reachable from $x$. Firstly, if $\alpha \in \B^c \setminus Q$, then $f_D(\alpha, x_D)$ reduces to the identity, hence $R_\alpha = \{ ( \alpha, z ) \}$. Secondly, if $\alpha = q^*$, then $f_D(\alpha, x_D)$ reduces to $g$, hence $R_{ q^* } = \{ ( q^* , b) : b \in \MostPermissive^*( z ) \}$. Thirdly, if $\alpha \in Q \setminus \{ q^* \}$, then $f_D(\alpha, x_D)$ reduces to the negation, hence $R_\alpha = \{ (\alpha, a) : a \in \B^{n-c} \}$. Therefore, $x$ reaches the following set of configurations:
\begin{align*}
    \MostPermissive^*( x ) &= \{ ( \alpha, z ) : \alpha \in \B^c \setminus Q \} \cup \{ (q^*, b) : b \in \MostPermissive^*( z ) \} \cup \{ (\alpha, a) : \alpha \in Q \setminus \{ q^* \}, a \in \B^{n-c} \},
\end{align*}
whence $| \MostPermissive^*( x ) | = (2^c - q) + r + 2^{n-c} (q - 1) = k$.
\end{proof}

\begin{corollary} \label{corollary:trapping_MP_commensurate}
 For any configuration $x$, we have
\[
    | \MostPermissive^*( x ) | \le | \TrappingGraph^*( x ) |, \qquad | \TrappingGraph^*( x ) | \le ( | \MostPermissive^*( x ) | )^2.
\]
Therefore, the most permissive and trapping update modes are commensurate.
\end{corollary}

We now exhibit three examples of min-trapspace configurations that are reachable in some update mode $\mu$ but not in another mode $\nu$; these three examples also show that $\mu$ and $\nu$ are not commensurate. Since the three examples follow a similar structure, we only provide a detailed explanation for the first example. We use the following shorthand notation: for all integers $i \le j$, let $i .. j = \{i, i+1, \dots, j\}$.

\begin{example}[Min-trapspace configurations reachable in history-based but not in cuttable] \label{example:H_C_not_commensurate}
Let $\HistoryNetwork \in \Functions(3)$ be the network from Example \ref{example:history_trajectory} and let $s = 101$. It is easily seen that there is no $\Cuttable$-trajectory starting at $x^0 = 000$ such that $s^a = s$. Then, for any $n \ge 5$, let $\hat{\HistoryNetwork} \in \Functions(n)$ be defined as
\[
    \hat{\HistoryNetwork}(x) = \begin{cases}
        ( \neg x_{123}, \neg x_{ 4 .. n-1 }, 1 ) & \text{if } x_n = 1 \\
        ( \HistoryNetwork( x_{123} ), 0_{ 4 .. n-1 }, 0 ) & \text{if } x_n = 0 \text{ and } x_{4 .. n-1} = 0 \text{ and } x_{123} \ne s \\
        ( \HistoryNetwork( s ), 0_{ 4 .. n-1 }, 1 ) & \text{if } x_n = 0 \text{ and } x_{4 .. n-1} = 0 \text{ and } x_{123} = s \\
        x & \text{otherwise}.
    \end{cases}
\]
Note that there is only one transition from the hyperplane $I = \{ x \in \B^n : x_n = 0 \}$ to its parallel $J = \{ x \in \B^n : x_n = 1 \}$. We note that $J$ is a minimal trapspace, as $T_f( y ) = J$ for all $y \in J$. Then the configuration $\hat{x} = 0^n$ can reach the whole of $J$ in $\History$ but it can only reach at most the seven configurations $\{ (y_{ 123 }, 0_{ 4 .. n-1 }, 0) : y \ne s \}$ in $\Cuttable$. Therefore, the history-based and cuttable update modes are not min-trapspace-equivalent.

This example also shows that the history-based and cuttable update modes are not commensurate. Indeed, we have
\[
    | \History^*( \hat{x} ) | \ge 2^{n-1}, \qquad | \Cuttable^*( \hat{x} ) | \le 7.
\]
Thus, there is no function $\psi : \N \to \N$ such that $| \History^*( x ) | \le \psi( | \Cuttable^*( x ) | )$ for all configurations $x$, as $\psi(7)$ would be unbounded. 

Finally, since $\MostPermissive^*( \hat{x} ) \supseteq \History^*( \hat{x} )$ and $\Interval^*( \hat{x} ) \subseteq \Cuttable^*( \hat{x} )$, this example shows more generally that the most permissive/history-based update mode is neither min-trap-equivalent nor commensurate to the cuttable/interval update mode.
\end{example}

\begin{example}[Min-trapspace configurations reachable in cuttable but not in history-based] \label{example:C_H_not_commensurate}
Let $\CuttableNetwork \in \Functions(3)$ be the network from Example \ref{example:cuttable_trajectory} and $s = 010$, then again no history-based or interval trajectory  starting at $x^0 = 000$ can have $s^a = s$. Then, for any $n \ge 5$, let $\hat{\CuttableNetwork} \in \Functions(n)$ be defined as
\[
    \hat{\CuttableNetwork}(x) = \begin{cases}
        ( \neg x_{123}, \neg x_{ 4 .. n-1 }, 1 ) & \text{if } x_n = 1 \\
        ( \CuttableNetwork( x_{123} ), 0_{ 4 .. n-1 }, 0 ) & \text{if } x_n = 0 \text{ and } x_{4 .. n-1} = 0 \text{ and } x_{123} \ne s \\
        ( \CuttableNetwork( s ), 0_{ 4 .. n-1 }, 1 ) & \text{if } x_n = 0 \text{ and } x_{4 .. n-1} = 0 \text{ and } x_{123} = s \\
        x & \text{otherwise}.
    \end{cases}
\]
This example is a ``mirror image'' to Example \ref{example:H_C_not_commensurate}. As such, this example shows that the most permissive/cuttable update mode is neither min-trap-equivalent nor commensurate to the history-based/interval update mode.
\end{example}


\begin{example}[Min-trapspace configurations reachable in interval but not in asynchronous] \label{example:I_A_not_commensurate}
Let $\IntervalNetwork \in \Functions(3)$ be the network from Example \ref{example:interval_trajectory} and $s = 011$, then there again no asynchronous trajectory starting at $x^0 = 000$ can have $s^a = s$. Then, for any $n \ge 5$, let $\hat{\IntervalNetwork} \in \Functions(n)$ be defined as
\[
    \hat{\IntervalNetwork}(x) = \begin{cases}
        ( \neg x_{123}, \neg x_{ 4 .. n-1 }, 1 ) & \text{if } x_n = 1 \\
        ( \IntervalNetwork( x_{123} ), 0_{ 4 .. n-1 }, 0 ) & \text{if } x_n = 0 \text{ and } x_{4 .. n-1} = 0 \text{ and } x_{123} \ne s \\
        ( \IntervalNetwork( s ), 0_{ 4 .. n-1 }, 1 ) & \text{if } x_n = 0 \text{ and } x_{4 .. n-1} = 0 \text{ and } x_{123} = s \\
        x & \text{otherwise}.
    \end{cases}
\]
Similarly to the previous examples, this example shows that the interval update mode is neither min-trapspace-equivalent nor commensurate to the asynchronous update mode.
\end{example}

\section{Consequences of using memory in update modes} \label{section:consequences}

In this subsection, we would like to highlight one important common point and three major differences between the asynchronous update mode $\Asynchronous$ and any other update mode with memory $\mu \in \{ \History, \TrappingGraph, \MostPermissive, \Subcube, \Interval, \Cuttable \}$ discussed in this paper. The key observation is that the memory of the trajectory offers a context upon which more transitions are allowed. Therefore, when we arrive at a particular configuration $x^{a-1}$, one may usually perform more transitions than what would be allowed if the trajectory were starting from there, as given by the asynchronous graph. This yields the following four features of update modes that use memory. In the sequel, we assume $\mu \in \{ \History, \TrappingGraph, \MostPermissive, \Subcube, \Interval, \Cuttable \}$ and we fix the Boolean network $f \in \Functions(n)$.

\begin{enumerate}
    \item \textbf{Trajectories closed by concatenation.}
Note that in a general update mode $\nu$, as defined at the start of Section \ref{subsection:unified_framework}, we may have two consecutive $\nu$-trajectories $x^0 \to_\nu \dots \to_\nu x^l$ and $x^l = y^0 \to_\nu \dots \to_\nu y^m$ such that their concatenation $x^0 \to \dots \to x^l = y^0 \to \dots \to y^m$ is no longer a $\nu$-trajectory. This happens for instance if the update mode $\nu$ constrains the source $s^a$ and the target $t^a$ to be equal to the starting configuration of the trajectory. However, for any $\mu \in \{ \Asynchronous, \History, \TrappingGraph, \MostPermissive, \Subcube, \Interval, \Cuttable \}$, trajectories are closed under concatenation: if $x^0 \to_\mu \dots \to_\mu x^l$ and $x^l = y^0 \to_\mu \dots \to_\mu y^m$, then $x^0 \to_\mu \dots \to_\mu x^l = y^0 \to_\mu \dots \to_\mu y^m$. Intuitively, if the transitions $y^{a-1} \to y^a$ are allowed with only the memory from $x^l = y^0$, then they will still be allowed with an enhanced memory that goes back to $x^0$.

\item \textbf{Trajectories not closed by suffix.}
Say a $\mu$-trajectory $(x^0, \dots, x^l)$ for some update mode $\mu$ is \Define{suffix-closed} if for all $1 \le a \le l$, $(x^a, \dots, x^l)$ is also a $\mu$-trajectory. Asynchronous trajectories are suffix-closed, i.e. if $x^0 \to_\Asynchronous x^1 \to_\Asynchronous \dots \to_\Asynchronous x^l$, then for all $1 \le a \le l$, $x^a \to_\Asynchronous \dots \to_\Asynchronous x^l$. This is no longer the case for any update mode $\mu$ with memory. In fact, any $\mu$-trajectory that is not an asynchronous trajectory (such a trajectory exists due to the hierarchy by trajectory in Figure \ref{figure:hierarchy}) is not suffix-closed, as it takes advantage of memory to use a transition that is not allowed in the asynchronous case. For instance, the $\History$-trajectory in Example \ref{example:history_trajectory} is not suffix-closed: the suffix $(111, 101)$ is not an $\History$-trajectory since $111$ is a fixed point.

\item \textbf{No transition graph.}
This item is closely related to the previous one. The asynchronous update mode could be simply represented by a transition graph, namely the asynchronous graph $\Asynchronous( f )$: $x^0 \to_\Asynchronous x^1 \to_\Asynchronous \dots \to_\Asynchronous x^l$ if and only if $( x^{a-1}, x^a )$ is an arc in $\Asynchronous( f )$. The same cannot be said for any other update mode discussed in this paper. Indeed, for any update mode with memory $\mu$, we have $x^0 \to_\mu x^1$ if and only if $(x^0, x^1)$ is an arc in $\Asynchronous( f )$; as such, the only possible choice for a ``$\mu$-transition graph'' (i.e. a transition graph according to $\mu$) would be $\Asynchronous( f )$. However, as shown in the hierarchy by trajectory, there exists a $\mu$-trajectory that is not asynchronous (see Example \ref{example:interval_trajectory} for instance); hence the $\mu$-trajectories cannot be represented by a transition graph.

\item \textbf{No clear definition of limit configurations and attractors.}
In the asynchronous update mode, there is a clear definition of an \Define{$\Asynchronous$-limit configuration} of $f$: those are the configurations $y$ such that, if $y \to_\Asynchronous^* x$, then $x \to_\Asynchronous^* y$. Equivalently, they belong to the terminal strong components of the asynchronous graph. In particular, any fixed point is a limit configuration for the asynchronous update mode. An \Define{$\Asynchronous$-attractor} of $f$ is any set of configurations of the form $\Asynchronous^*( y )$, where $y$ is a limit configuration; equivalently it is a terminal strong component of the asynchronous graph.

One can generalise these definitions to any update mode $\mu$: say $y$ is a \Define{$\mu$-limit configuration} of $f$ if $y \to_\mu^* x$ implies $x \to_\mu^* y$, and a \Define{$\mu$-attractor} of $f$ is any set of configurations of the form $\mu^*( y )$, where $y$ is a $\mu$-limit configuration. These definitions are mathematically valid, and indeed for any fixed point $y$, $y$ is a $\mu$-limit configuration and $\{ y \}$ is a $\mu$-attractor. Moreover, Lemma \ref{lemma:min-trapspace_configurations} yields the following characterisation of $\mu$-attractors for $\mu \in \{ 
\MostPermissive, \TrappingGraph, \Subcube \}$.

\begin{proposition} \label{proposition:mu-limit_configurations}
Let $f \in \Functions( n )$ and $\mu \in \{ \MostPermissive, \TrappingGraph, \Subcube \}$. Then the $\mu$-limit configurations of $f$ are the min-trapspace configurations and the $\mu$-attractors of $f$ are its minimal trapspaces.
\end{proposition}

\begin{proof}
Firstly, if $y$ is a min-trapspace-configuration and $z \in T_f( y )$, then $T_f( z ) \subseteq T_f( y )$, hence $T_f( z ) = T_f( y )$ by minimality, and $z$ is also a min-trapspace configuration.

Now, suppose $y$ is a min-trapspace-configuration and $z \in \mu^*( y )$. We then have $z \in T_f( y )$ and by Lemma \ref{lemma:min-trapspace_configurations}, $\mu^*( z ) = T_f( z ) = T_f( y )$. Thus $y \in \mu^*( z )$ and $y$ is a $\mu$-limit configuration. 

Conversely, if $x$ is not a min-trapspace, then by Lemma \ref{lemma:min-trapspace_configurations} there exists a min-trapspace configuration $y$ such that $y \in \mu^*( x )$. However, since $T_f( y ) \subsetneq T_f( x )$, we have $x \notin \mu^*( y )$. Thus, $x$ is not a $\mu$-limit configuration.

Finally, since $\mu^*( y ) = T_f( y )$ for any min-trapspace configuration, any minimal trapspace is a $\mu$-attractor and vice versa.
\end{proof}

However, an undesirable aspect of that definition of limit configurations is that there is no neat hierarchy by limit configurations that reflects the power of using memory, akin to what we have obtained for reachability and trajectories. Indeed, a configuration $y$ can be a $\mu$-limit but not a $\Asynchronous$-limit, and vice versa. 

We illustrate this situation via the following two examples. For each, we give three graphs (omitting loops, as usual): the asynchronous graph $\Asynchronous$ in blue, the asynchronous reachability graph $\Asynchronous^*$ with arcs $x \to_\Asynchronous^* y$ in red, and the history reachability graph $\History^*$ with arcs $x \to_\History^* y$ in orange. In the first example, $00$ is an $\Asynchronous$-limit configuration but not an $\History$-limit configuration (as $00 \to_\History^* 11$ but $11 \not\to_\History^* 00$).

\begin{tikzpicture} 

    \begin{scope}[xshift = 0cm]
        \node[blue] (A) at (1,3) {$\Asynchronous$};
        \node (00) at (0,0) {$00$};
        \node (01) at (0,2) {$01$};
        \node (10) at (2,0) {$10$};
        \node (11) at (2,2) {$11$};
    
        \path[draw] (00) -- (01) -- (11)
        (00) -- (10) -- (11);
        
        \draw[thick,latex-latex, blue] (00) -- (10);
        \draw[thick,latex-latex, blue] (00) -- (01);
    \end{scope}

    \begin{scope}[xshift = 5cm]
        \node[red] (A) at (1,3) {$\Asynchronous^*$};
        \node (00) at (0,0) {$00$};
        \node (01) at (0,2) {$01$};
        \node (10) at (2,0) {$10$};
        \node (11) at (2,2) {$11$};
    
        \path[draw] (00) -- (01) -- (11)
        (00) -- (10) -- (11);
        
        \draw[thick,latex-latex, red] (00) -- (10);
        \draw[thick,latex-latex, red] (00) -- (01);
        \draw[thick,latex-latex, red] (01) -- (10);
    \end{scope}

    \begin{scope}[xshift = 10cm]
        \node[orange] (H) at (1,3) {$\History^*$};
        \node (00) at (0,0) {$00$};
        \node (01) at (0,2) {$01$};
        \node (10) at (2,0) {$10$};
        \node (11) at (2,2) {$11$};
    
        \path[draw] (00) -- (01) -- (11)
        (00) -- (10) -- (11);
        
        \draw[thick,latex-latex, orange] (00) -- (10);
        \draw[thick,latex-latex, orange] (00) -- (01);
        \draw[thick,latex-latex, orange] (01) -- (10);
        \draw[thick,-latex, orange] (00) -- (11);
        \draw[thick,-latex, orange] (01) -- (11);
        \draw[thick,-latex, orange] (10) -- (11);
    \end{scope}
    
\end{tikzpicture}

In the second example, $11$ is an $\History$-limit configuration but not an $\Asynchronous$-limit configuration (as $11 \to_\Asynchronous^* 01$ but $01 \not\to_\Asynchronous^* 11$).

\begin{tikzpicture} 

    \begin{scope}[xshift = 0cm]
    \node[blue] (A) at (1,3) {$\Asynchronous$};
        \node (00) at (0,0) {$00$};
        \node (01) at (0,2) {$01$};
        \node (10) at (2,0) {$10$};
        \node (11) at (2,2) {$11$};
    
        \path[draw] (00) -- (01) -- (11)
        (00) -- (10) -- (11);
        
        \draw[thick,latex-latex, blue] (00) -- (10);
        \draw[thick, -latex, blue] (11) -- (01);
        \draw[thick,latex-latex, blue] (00) -- (01);
        \draw[thick, -latex, blue] (11) -- (10);
    \end{scope}

    \begin{scope}[xshift = 5cm]
    \node[red] (A) at (1,3) {$\Asynchronous^*$};
        \node (00) at (0,0) {$00$};
        \node (01) at (0,2) {$01$};
        \node (10) at (2,0) {$10$};
        \node (11) at (2,2) {$11$};
    
        \path[draw] (00) -- (01) -- (11)
        (00) -- (10) -- (11);
        
        \draw[thick,latex-latex, red] (00) -- (10);
        \draw[thick, -latex, red] (11) -- (01);
        \draw[thick,latex-latex, red] (00) -- (01);
        \draw[thick, -latex, red] (11) -- (10);
        \draw[thick,latex-latex, red] (01) -- (10);
        \draw[thick, -latex, red] (11) -- (00);
    \end{scope}

    \begin{scope}[xshift = 10cm]
    \node[orange] (A) at (1,3) {$\History^*$};
        \node (00) at (0,0) {$00$};
        \node (01) at (0,2) {$01$};
        \node (10) at (2,0) {$10$};
        \node (11) at (2,2) {$11$};
    
        \path[draw] (00) -- (01) -- (11)
        (00) -- (10) -- (11);
        
        \draw[thick,latex-latex, orange] (00) -- (10);
        \draw[thick, latex-latex, orange] (11) -- (01);
        \draw[thick,latex-latex, orange] (00) -- (01);
        \draw[thick, latex-latex, orange] (11) -- (10);
        \draw[thick,latex-latex, orange] (01) -- (10);
        \draw[thick, latex-latex, orange] (11) -- (00);
    \end{scope}
    
\end{tikzpicture}

We make a final remark about attractors. Asynchronous attractors satisfy the following absorption property: if $x^0 \to_\Asynchronous \dots \to_\Asynchronous x^l$ and $x^a$ belongs to the $\Asynchronous$-attractor $A$ for some $0 \le a \le l-1$, then $x^{a+1}, \dots, x^l$ all belong to $A$ as well. However, our definition of $\mu$-attractors does not preserve the absorption property. Some trajectories (such as the one in Example \ref{example:history_trajectory}) can even go through fixed points! As such, the term ``attractor'' may be misleading when it comes to update modes with memory.


    
        

\end{enumerate}

\section{Discussion}\label{section:discussion}

We proposed a novel and unifying characterization of BN dynamics that can exploit a memory along
trajectories in order to reach configurations that are not reachable with conventional (general)
asynchronous updates.
We have shown how update modes of the literature such as the most permissive, the interval, and the
cuttable modes can be expressed in this framework, emphasizing the notion of memory they employ.
Besides existing modes, we took advantage of this framework to introduce three novel update modes:
the history-based mode, which can use any past configuration as a reference to update its current
one, and the trapping and subcube-based modes, which can also revive and update a past configuration.
The comparison of trajectories and reachable configurations between these memory-based update modes
resulted in a hierarchy of (weak) simulation, in which the trapping mode, and its equivalent
subcube-based mode, subsume all others, including the most permissive mode.
Moreover, trapping and most permissive modes coincide on the reachability of configurations in minimal trapspaces, which in these cases, are the attractors of the dynamics.

In \cite{automata21,PS22}, authors also provided a unifying characterization of these unconventional
update modes as non-deterministic updates of components, in contrast with standard block-sequential
or (a)synchronous updates, which are deterministic.
These non-deterministic updates reflect some sort of strong asynchronicity within the computation of the
update of one component, enabling to generate more trajectories.
Here, we brought a different view on these modes, by focusing on a notion of memory used to compute
next configurations from the knowledge of past ones.
In addition to a complementary understanding of modes having larger dynamics than the general
asynchronous, our framework can ease and motivate the definition of further update modes to be
explored.

The update modes and hierarchy we considered here are based on a fully asynchronous definition of the trajectories, where only one component is updated at a time. In future work, these definitions could be extended to synchronous and general asynchronous modes where all or subsets of components are updated simultaneously.
Finally, the results on trapping, subcube-based, and most permissive modes call for further study on the combinatorics and structure of BN trapspaces.

\paragraph{Acknowledgments}
Work of LP was supported by the French Agence Nationale pour la Recherche (ANR)
in the scope of the project ``BNeDiction'' (grant number ANR-20-CE45-0001).
Work of SR was supported by the French Agence Nationale pour la Recherche (ANR) in the scope of the project ``REBON'' (grant number ANR-23-CE45-0008).

\bibliographystyle{plain}
\bibliography{trapspaces.bib}

\appendix

\section{Proofs of equivalence of memory-based definitions}

We prove the equivalence of definitions of memory-based trajectories with the original
iteration-based definitions (Sect.~\ref{sec:updates}).

\subsection{Most permissive updates}\label{seq:eq-mp}

We follow the notations of Sect.~\ref{sec:mp}.

Let us consider any memory-based most permissive trajectory $(x^0, \ldots, x^l)$.
Let us define $\tilde x^0=x^0$, and for any $1 \leq a \leq l$, $\tilde x^a\in\mathbb P^n$ such that
$\Delta(\tilde x^{a-1},\tilde x^{a}) = \Delta(x^{a-1},x^a) = \{i^a\}$, and $\tilde x^a_{i^a} = \nearrow$ if $x^a=1$ and
$\tilde x^a_{i^a} = \searrow$ if $x^a=0$.
Then, remark that for any $0 \leq a \leq l$, $[x^0,\ldots,x^a]=\gamma(\tilde x^a)$.
Thus, $\tilde x^0 \to \ldots \to \tilde x^l$ is a valid sequence of most permissive iterations.

Conversely, let $x^0 \to \tilde x^1 \to \ldots \to \tilde x^l$ be any sequence of most permissive
iterations, with $\tilde x^1,\ldots,\tilde x^l\in\mathbb P^n$.
Given $\tilde x\in\mathbb P^n$, let us define the configuration $\beta(\tilde x) \in\mathbb B^n$ as 
$\beta(\tilde x)_i = 1$ whenever $\tilde x_i\in \{1,\nearrow\}$ and 
$\beta(\tilde x)_i = 0$ otherwise.
Remark that for any $k\in \{0,\ldots,l\}$, $\gamma(\tilde x^k)\subseteq [x^0,\beta(\tilde x^1),\ldots,\beta(\tilde x^k)]$.
Thus, there exists a corresponding memory-based trajectory $(x^0, \ldots, \beta(\tilde x^l))$,
where we filter out iterations that change a component state from $\nearrow$ (resp. $\searrow$) to
$1$ (resp. $0$) to avoid successive repetition of a configuration.
Formally, the corresponding memory-based trajectory is $(x^0,\beta(\tilde
x^{\pi^{-1}(1)}),\cdots,\beta(\tilde x^{\pi^{-1}(l')}))$,
with $\pi : \{0,\cdots,l\} \to \{0,\cdots,l\}$ 
such that
$\pi(0) = 0$, and for any $b\in \{1,\ldots,l\}$, 
$\pi(b) = \pi(b-1)$ whenever $\beta(\tilde x^b)=\beta(\tilde x^{b-1})$
and $\pi(b) = \pi(b-1)+1$ otherwise,
and $l'=\pi(l)$.

\subsection{Interval updates}\label{sec:eq-interval}

We follow the notations of Sect.~\ref{sec:interval}.

Let us consider any memory-based interval trajectory $(x^0, \ldots, x^l)$ with associated $V^a$
vectors and $s^a$ configurations for $0\leq a\leq l$, considering $s^0=x^0$.
Intuitively, to obtain the corresponding interval iterations, 
we will map the $x^a$ configurations to the \emph{write} nodes, and $s^a$
configurations to the \emph{read} nodes:
given any pair of configurations $x,s\in\mathbb B^n$, let us define the configuration
$z(x,s)\in\mathbb B^{2n}$ as follows: for all $i\in[n]$,
    $z(x,s)_{2i-1} = x_i$ and $z(x,s)_{2i} = s_i$.
Given any $1 \leq a < l$, we show there is a sequence of asynchronous iterations of $\tilde f$ from
the configuration $z(x^{a-1},s^{a-1})$ to the configuration $z(x^a,s^a)$.
The idea is to first update the \emph{read} nodes:
for any component $i\in\Delta(s^{a-1},s^a)$, by definition of $s^a$ and $V^a$,
$s^a_i = x^{V^a_{i}}_i = x^{a-1}$.
Therefore, $f_{2i}(z(x^{a-1},s^{a-1})) = x^{a-1}$.
By updating only components in $\Delta(s^{a-1},s^a)$, 
there is a sequence of iterations from $z(x^{a-1},s^{a-1})$ to $z(x^{a-1},s^a)$.
Let $i^a$ be the component updated between $x^{a-1}$ and $x^a$.
Because $V^a_{i^a}=a-1$, $x^{a-1}=s^a$.
Thus, $\tilde f_{2i^a-1}(z(x^{a-1},s^a)) = f_{i^a}(s^a) = x^a_{i^a}$.
Hence, $z(x^{a-1},s^a)\to z(x^a,s^a)$.

Conversely, given a sequence of interval iterations $z(x^0,x^0) \to z^1 \to \ldots \to z^l$, one can
compute a corresponding memory-based trajectory $(x^0,x^1=\tau_{\text
w}(z^{\pi^{-1}(1)}),\ldots,x^{l'}=\tau_{\text w}(z^{\pi^{-1}(l')}))$ corresponding on the projection
over the write nodes,
with
$\tau_{\text w} : \B^{2n} \to \B^n$ maps configurations of the interval semantics to configurations
of the BN $f$ by projecting on the \emph{write} nodes, i.e., $\tau_{\text w}(z)_i = z_{2i-1}$ for
every $i \in [n]$;
and
with $\pi : \{0,\cdots,l\} \to \{0,\cdots,l\}$ such that
$\pi(0) = 0$ and for any $b\in \{1,\ldots,l\}$, 
$\pi(b) = \pi(b-1)$ whenever $\tau_{\text w}(z^b)=\tau_{\text w}(z^{b-1})$
and $\pi(b) = \pi(b-1)+1$ otherwise,
and $l'=\pi(l)$.
The corresponding vectors $V^a$ for $1\leq a\leq l'$ are defined such that 
$s^a_j = x^{V^a_j}_j = \tau({z^{\pi^{-1}(a)}})_j$ for all $j\in[n]$.

\subsection{Cuttable updates}\label{sec:eq-cuttable}

We follow the notations of Sect.~\ref{sec:cuttable}.

The correspondence between memory-based trajectories and sequences of cuttable iterations follows
the same scheme as in the previous section:
between two configurations $x^{a-1}$, $x^a$ of a memory-based trajectory,
there exists a corresponding sequence of iterations that first update the read nodes according to
$C^a$, and then update the write node corresponding to the component $i^a$.
Conversely, given a sequence of cuttable iterations, one obtain a corresponding memory-based
trajectory by projecting over the write nodes.

Let us consider any memory-based cuttable trajectory $(x^0, \ldots, x^l)$ with associated $C^a$
matrices and $s^a$ configurations for $0\leq a\leq l$, considering $C^0=0^{n\times n}$ and $s^0=x^0$.
Given a configuration $x\in\mathbb B^n$ of the BN $f$, a trajectory $(x^0, \ldots, x^a)$, and any
matrix $C \leq (a-1)^{n\times x}$, let us define the
configuration
$z(x,(x^0,\ldots,x^a),C)\in\mathbb B^{n+|E|}$  of the $E$-extension of $f$ as follows: for all $i\in[n]$,
    $z(x,C)_i = x_i$, and
    for all $(j,k)\in E$,
    $z(x,C)_{(j,k)} = x^{C_{k,j}}_j$.
Given any $1 \leq a < l$, we show there is a sequence of asynchronous iterations of $\tilde f$, the
$E$-extension of $f$, from the configuration $z(x^{a-1},C^{a-1})$ to the configuration $z(x^a,C^a)$.
The idea is to first update the \emph{read} nodes, i.e., all the nodes $(j,k)\in [n]\times[n]$
such that $C^a_{k,j} \neq C^{a-1}_{k,j}$:
by definition, $\tilde f_{(j,k)}(z(x^{a-1},C^{a-1})) = x_j^{C^{a-1}_{k,j}}$.
This leads to a sequence of iterations from
$z(x^{a-1},C^{a-1})$ to $z(x^{a-1},C^a)$.
Let $i^a$ be the component updated between $x^{a-1}$ and $x^a$: we obtain that
$s^a = \pi^{i^a}(x^{a-1})$, thus $\tilde f_{i^a}(z(x^{a-1},C^a)) = x^a_{i^a}$.
Hence, $z(x^{a-1},C^a)\to z(x^a,C^a)$.

Conversely, given a sequence of iterations $z(x^0,C^0) \to z^1 \to \ldots \to z^l$ of the
$E$-extension BN $\tilde f$, one can
compute a corresponding memory-based trajectory $(x^0,x^1=\tau_{\text
w}(z^{\pi^{-1}(1)}),\ldots,x^{l'}=\tau_{\text w}(z^{\pi^{-1}(l')}))$ corresponding on the projection
over the write nodes,
with
$\tau_{\text w} : \B^{n+|E|} \to \B^n$ maps configurations of $\tilde f$ to configurations
of $f$ by projecting on the \emph{write} nodes, i.e., $\tau_{\text w}(z)_i = z_{i}$ for
every $i \in [n]$;
and
with $\pi : \{0,\cdots,l\} \to \{0,\cdots,l\}$ such that
$\pi(0) = 0$ and for any $b\in \{1,\ldots,l\}$, 
$\pi(b) = \pi(b-1)$ whenever $\tau_{\text w}(z^b)=\tau_{\text w}(z^{b-1})$
and $\pi(b) = \pi(b-1)+1$ otherwise,
and $l'=\pi(l)$.
The corresponding matrices $C^a$ for $1\leq a\leq l'$ are defined such that 
$x^{C^a_{k,j}}_j = z^{\pi^{-1}(a)}_{(j,k)}$ for all $j,k\in[n]$.

\end{document}